\documentclass{llncs}

\newcommand{\knip}[1]{}
\usepackage{url}

\usepackage{graphicx}

\newcommand{\NP}{{\sf NP}}
\newcommand{\W}{\ensuremath{\mathsf{W}}}

\newcommand{\etal}{et al.~}

\newcommand{\ie}{i.e.~}
\newcommand{\eg}{e.g.~}

\newcommand{\marge}[1]{}

\DeclareSymbolFont{AMSb}{U}{msb}{m}{n}
\DeclareSymbolFontAlphabet{\mathbb}{AMSb}

\newcommand{\problemIDP}{\textsc{$k$-Induced Disjoint Paths}}

\newcommand{\problemIAP}{\textsc{$k$-in-a-Path}}
\newcommand{\problemIAT}{\textsc{$k$-in-a-Tree}}
\newcommand{\problemIAC}{\textsc{$k$-in-a-Cycle}}

\begin{document}

\title{Induced Disjoint Paths in Claw-Free Graphs\thanks{This work is supported by EPSRC (EP/G043434/1) and Royal Society (JP100692).
The research leading to these results has also received funding from the European Research Council under the European Union's Seventh Framework Programme (FP/2007-2013) / ERC Grant Agreement n. 267959.
A preliminary version of this paper appeared as an extended abstract in the proceedings of ESA 2012~\cite{GPV12b}.
}}
\author{
Petr A.~Golovach\inst{1}
\and Dani\"el Paulusma\inst{2}
\and Erik Jan van Leeuwen\inst{3}
\institute{Department of informatics, University of Bergen, Norway\\
\email{petr.golovach@ii.uib.no}
\and
School of Engineering and Computer Science, Durham University, UK\\ 
\email{daniel.paulusma@durham.ac.uk}
\and 
Max-Planck Institut f\"{u}r Informatik, Saarbr\"{u}cken, Germany,
\email{erikjan@mpi-inf.mpg.de}
}}
\maketitle
\thispagestyle{plain}

\begin{abstract}
Paths $P_1,\ldots,P_k$ in a graph $G=(V,E)$ are said to be mutually induced if 
for any $1\leq i<j \leq k$, $P_i$ and $P_j$  have neither common vertices nor adjacent vertices (except perhaps their end-vertices). 
The {\sc Induced Disjoint Paths} problem is to test whether 
a graph $G$ with $k$ pairs of specified vertices $(s_i,t_i)$ contains 
$k$ mutually induced paths $P_i$ such that $P_i$ connects $s_i$ and $t_i$ for $i=1,\ldots,k$.
We show that this problem is fixed-parameter tractable for claw-free graphs when parameterized by $k$.
Several related problems, such as the {\problemIAP} problem, are 
 proven
to be fixed-parameter tractable for claw-free graphs as well.
We show that an improvement of these results in certain directions is unlikely, for example by noting that the 
{\sc Induced Disjoint Paths} problem cannot have a polynomial kernel for line graphs (a  type
of claw-free graphs), unless \NP\ $\subseteq$\ co$\NP/$poly. Moreover, the problem becomes 
\NP-complete, even when $k=2$, for the more general class of $K_{1,4}$-free graphs.
Finally, we show that the $n^{O(k)}$-time algorithm of Fiala \etal for testing whether a claw-free graph contains some $k$-vertex graph $H$ as a topological induced minor is essentially optimal by proving that this problem is \W[1]-hard even if $G$ and $H$ are line graphs.
\end{abstract}

\section{Introduction}\label{s-intro}
The problem of finding disjoint paths of a certain type in a graph has received considerable attention in recent years. The regular {\sc Disjoint Paths} problem is to test whether a graph $G$ with $k$ pairs of specified vertices $(s_i,t_i)$ contains a set of $k$ mutually vertex-disjoint paths $P_1,\ldots,P_k$ 
such that $P_i$ has end-vertices $s_i$ and $t_i$ for $i=1,\ldots,k$.
The subgraph of $G$ induced by the vertices of these paths is called a {\it linkage}.
This problem is included in Karp's list of \NP-compete problems~\cite{Ka75}, provided that $k$ is part of the input.
If $k$ is any {\it fixed} integer, that is, not part of the input, the problem is called {\sc $k$-Disjoint Paths} and can be solved in $O(n^3)$ time for $n$-vertex graphs, as shown by  Robertson and Seymour~\cite{RS95} in one of their keystone papers on graph minor theory.

In this paper, we study a generalization of the {\sc Disjoint Paths} problem by considering its \emph{induced} version.
We say that
paths $P_1,\ldots, P_k$ in a graph $G=(V,E)$  are {\it mutually induced} if 
for any $1\leq i<j \leq k$, $P_i$ and $P_j$  have neither common vertices,  
\ie $V(P_i)\cap V(P_j)=\emptyset$,
nor adjacent vertices, \ie $uv\notin E$ for any $u\in V(P_i), v\in V(P_j)$, except perhaps their end-vertices. 
The subgraph of $G$ induced by the vertices of such paths is called an {\it induced linkage}.
We observe that the paths $P_1,\ldots, P_k$ are not required to be induced paths in $G$. However, this may be assumed without loss of generality, because we can
replace non-induced paths by short\-cuts. 

We can now define the following problem, where we call the vertex pairs specified in the input {\it terminal pairs} and their vertices {\it terminals}.

\ \\[-0.3cm]
\noindent
{\sc Induced Disjoint Paths}\\
{\it Instance:} a graph $G$ with $k$ terminal pairs $(s_i,t_i)$ for
$i=1,\ldots,k$.\\
{\it Question:} does $G$ contain $k$ mutually induced paths $P_i$ such
that $P_i$ connects\\ \hspace*{1.5cm} terminals $s_i$ and $t_i$ for $i=1,\ldots,k$?

\ \\[-0.3cm]
\noindent
When $k$ is fixed, we call this the {\problemIDP} problem.

Observe that the {\sc Induced Disjoint Paths} problem can indeed be seen as a generalization of the {\sc Disjoint Paths} problem, since the latter can be reduced to the former by subdividing every edge of the graph. This generalization makes the problem significantly harder. In contrast to the original, non-induced version, the {\sc $k$-Induced Disjoint Paths} problem is \NP-complete even for $k=2$~\cite{Bi91,Fe89}.

The hardness of the {\problemIDP} problem motivates an investigation into graph classes for which it may still be tractable. 
Below, we briefly survey existing results.

\subsection{Known Results for Special Graph Classes}\label{s-known}

 For planar graphs, {\sc Induced Disjoint Paths} stays \NP-complete, as it generalizes {\sc Disjoint Paths} for planar graphs, which is \NP-complete as shown by
Lynch~\cite{Ly75}.
However,  Kobayashi and Kawarabayashi~\cite{KK12} presented
an algorithm for {\problemIDP} on planar graphs that runs in linear time for any fixed $k$, 
improving on an earlier algorithm by Reed, Robertson, Schrijver and Seymour~\cite{RRSS93}.
For AT-free graphs~\cite{GPV12} and chordal graphs~\cite{BGHHKP12},  {\sc Induced Disjoint Paths} is polynomial-time solvable,
whereas the problem is linear-time solvable for circular-arc graphs~\cite{GPV13}.

For claw-free graphs (graphs where no vertex has three pairwise nonadjacent neighbors), Fiala~\etal\cite{FKLP12} 
showed that the {\sc Induced Disjoint Paths} problem is \NP-complete.
They showed that this holds even for line graphs, a subclass of the class of claw-free graphs. They also gave a polynomial-time algorithm for {\problemIDP} for any fixed $k$. Their approach is based on a modification of the claw-free input graph to a special type of claw-free graph, namely to a quasi-line graph, in order to use the characterization of quasi-line graphs  by Chudnovsky and Seymour~\cite{CS05}. 
This transformation may require $\Omega(n^{2k})$  time due to some 
brute-force guessing, in particular as claw-freeness must be preserved. 

\subsection{Related Problems}\label{s-related}

A study on induced disjoint paths can also be justified from another direction, one that focuses on detecting induced subgraphs such as cycles, paths, and trees that contain some set of $k$ specified vertices, which are also called terminals. The corresponding decision problems are called {\problemIAC}, 
{\problemIAP}, and {\problemIAT}, respectively.
These problems are closely related to each other and to the {\sc $k$-Induced Disjoint Paths} problem. 

For general graphs, even the problems  {\sc $2$-in-a-Cycle} and {\sc $3$-in-a-Path} are \NP-complete~\cite{Bi91,Fe89}, whereas the {\problemIAT} problem is polynomial-time solvable for $k=3$~\cite{CS10}, open for any fixed $k \geq 4$, and \NP-complete when $k$ is part of the input~\cite{derhy2009}. Several polynomial-time solvable cases are known for graph classes, see \eg\cite{DPT09,GPV12,KK09,LT10,RRSS93}.
For claw-free graphs, the problems {\problemIAT} and {\problemIAP} are equivalent and polynomial-time solvable for any fixed integer $k$~\cite{FKLP12}. Consequently, the same holds for the {\problemIAC} problem~\cite{FKLP12}. 

As a final motivation for our work, we note that just as disjoint paths are important for  (topological) graph minors, one may hope that induced disjoint paths are useful for finding induced (topological) minors in polynomial time on certain graph classes.
Whereas the problems of
detecting whether a graph contains some fixed graph $H$ as a minor or topological minor can be solved in cubic time for any fixed graph $H$~\cite{GKMW11,RS95}, the
complexity classifications of both problems with respect to some fixed graph $H$ as induced minor or induced topological minor are still wide open.
So far, only partial results~\cite{FKMP95,FKP12b,LLMT09}, which consist of both polynomial-time solvable and \NP-complete cases, are known for the latter two problems on general graphs. 
In contrast, Fiala \etal\cite{FKLP12} use  their algorithm for
the {\sc $k$-Induced Disjoint Paths} problem  to obtain an $n^{O(k)}$-time algorithm that solves 
the problem of testing whether a claw-free graph $G$ on $n$ vertices contains a graph $H$ on $k$ vertices as a topological induced minor. 
This problem is also called the {\sc Induced Topological Minor} problem.
For fixed graphs $H$, the {\sc Induced Topological Minor}  problem
is also known to be polynomial-time solvable for AT-free graphs~\cite{GPV12} and chordal graphs~\cite{BGHHKP12}.
When parameterized by $|V(H)|$, it has been proven to be \W[1]-hard for cobipartite graphs~\cite{GPV12} (which form a subclass of AT-free graphs)
and for split graphs~\cite{GKPT12} (which form a subclass of chordal graphs).

\subsection{Our Results}\label{s-ours}

We provide new insights into the computational complexity of the {\sc Induced Disjoints Paths} problem and the related problems  {\problemIAP} and 
{\sc Induced Topological Minor} for claw-free and line graphs.

In Section~\ref{s-fpt} we improve on the aforementioned result of Fiala \etal\cite{FKLP12} by showing that  
{\sc Induced Disjoint Paths}
is  fixed-parameter tractable on claw-free graphs when parameterized by the number of terminal pairs $k$, that is,
can be solved in time~$f(k)n^{O(1)}$ on $n$-vertex claw-free graphs with $k$ terminal pairs, where $f$ is
some computable function $f$ that only depends on $k$.
Our approach circumvents the time-consuming transformation to quasi-line graphs of Fiala~\etal\cite{FKLP12}, and is based on an algorithmic application of the characterization for claw-free graphs by Chudnovsky and Seymour. Hermelin~\etal\cite{HMVW11} recently applied such an algorithmic structure theorem to \textsc{Dominating Set} on claw-free graphs. 
However, their algorithm reduces the strip-structure to have size polynomial in $k$ and then follows an exhaustive enumeration strategy. For {\problemIDP}, such an approach seems unsuitable, and our arguments thus differ substantially from those in~\cite{HMVW11}.

In Section~\ref{s-fpt} we also prove that the problems {\problemIAP} (or equivalently {\problemIAT}) and {\problemIAC} are fixed-parameter tractable when parameterized by $k$. This gives some answer to an open question of Bruhn and Saito~\cite{BS12}. They gave necessary and sufficient conditions for the existence of a path through three given vertices in a claw-free graph and
asked whether such conditions also exist for {\problemIAP} with $k\geq 4$.  However, as this problem is \NP-complete even for line graphs when $k$ is part of the input~\cite{FKLP12}, 
showing that it is fixed-parameter tractable may be the best answer 
we can hope for. 

Recall that Fiala \etal\cite{FKLP12}  gave an $n^{O(k)}$-time algorithm for testing
whether a claw-free graph $G$ on $n$ vertices contains a graph $H$ on $k$ vertices as a topological induced minor. 
We  prove in Section~\ref{s-con} that 
{\sc Induced Topological Minor} is  \W[1]-hard 
when parameterized by $|V(H)|$, 
even if $G$ and $H$ are line graphs.
This means that this problem is unlikely to be fixed-parameter tractable for this graph class.

In Section~\ref{s-con} we also show that our results for the 
{\sc Induced Disjoint Paths}
problem for claw-free graphs are best possible in the following ways. 
First, we show that the 
problem does not allow a polynomial kernel even for line graphs, unless \NP\ $\subseteq$\ co\NP$/$poly.
Second, we observe that a result from Derhy and Picouleau~\cite{derhy2009} immediately implies that {\sc $2$-Induced Disjoint-Paths} is \NP-complete on $K_{1,4}$-free graphs (graphs in which no vertex has four pairwise nonadjacent neighbors). 
We also state some related open problems in this section.

\section{Preliminaries}\label{s-pre}

{\bf Basic Graph Terminology.}
We only consider finite undirected graphs that have no loops and no multiple edges.
We refer to the text-book of Diestel~\cite{Di05} for any  standard graph terminology not used in our paper.  

For a subset $S\subseteq V$, the graph $G[S]$ denotes the subgraph of $G=(V,E)$ {\it induced by} $S$, that is, 
the graph with vertex set $S$ and edge set $\{uv \in E \mid  u,v\in S\}$.  We write $G-S=G[V\setminus S]$. For a vertex $u$ and a subgraph $F$ of $G$ that does not contain $u$ we write $F+u=G[V_F\cup \{u\}]$. We call the vertices $v_1$ and $v_r$ of a path $P=v_1\cdots v_r$ the {\it ends} or {\it end-vertices} of $P$. 

An {\it independent set} in a graph $G$ is a set of vertices that are mutually non-adjacent.  
We denote the maximum size of an independent set in a graph $G$ by $\alpha(G)$.
A {\it clique} in a graph $G$ is a set of vertices that are mutually adjacent.

Let $G=(V,E)$ be a graph. We denote the (open) neighborhood of a vertex $u$ by $N_G(u)=\{v \mid uv\in E\}$ and its closed neighborhood by $N_G[u]=N_G(u)\cup \{u\}$. We denote the  neighborhood of a set $U\subseteq V$ by $N_G(U)=\{v\in V\setminus U \mid  uv\in E\; \mbox{for some}\; u\in U\}$, 
and $N_G[U]=U\cup N_G(U)$. We omit indices if it does not create confusion.

\medskip
\noindent
{\bf Graph Operations and Containment Relations.}
Let $e=uv$ be an edge in a graph $G$. The \emph{edge contraction} of $e$
removes $u$ and $v$ from $G$, and replaces them by a new vertex adjacent to
precisely those vertices to which $u$ or $v$ were adjacent. 
In the case that one of the
two vertices, say $u$, has exactly two neighbors that in addition are
nonadjacent, then we call this operation the \emph{vertex dissolution} of
$u$.

Let $G$ and $H$ be two graphs. Then $G$ contains $H$ as an {\it induced minor} or {\it induced topological minor} if $G$ can be modified into $H$ by 
a sequence of edge contractions and vertex deletions, or vertex dissolutions and vertex deletions, respectively.

The problems {\sc Induced Minor} and {\sc Induced Topological Minor} are to test whether a graph $G$ contains a graph $H$ as an induced minor or induced topological minor, respectively.  Both problems are \NP-complete even when $G$ and $H$ are restricted to be line graphs~\cite{FKP12}. Hence, it is natural to study the computational complexity after excluding the graph $H$ from the input; when $H$ is fixed the problems are denoted {\sc $H$-Induced Minor} and {\sc $H$-Induced Topological Minor}. 

 The \emph{edge subdivision} operation replaces an edge $vw$ in a graph $G$ by a new vertex $u$ with edges $uv$ and $uw$.
A graph $H'$ is a  {\it subdivision} of a graph $H$ if $H$ can be modified into $H'$ by a sequence of edge subdivisions. 
We note that an edge subdivision is the ``dual'' operation of a vertex dissolution.
Hence, a graph $G$ contains a graph $H$ as an induced topological minor if and only if $G$ contains an induced subgraph that is isomorphic to a subdivision of $H$.
This alternative definition brings us to the following variant. Let $G$ be a graph in which we specify $k$ distinct vertices ordered as $u_1,\ldots,u_k$. Let $H$ be a $k$-vertex graph, the vertices of which are ordered as $x_1,\ldots,x_k$.
Then $G$ contains $H$ as an induced topological minor {\it anchored} in $u_1,\ldots,u_k$ if
$G$ contains an induced subgraph isomorphic to a subdivision of $H$ such that the isomorphism maps $u_i$ to $x_i$ for $i=1,\ldots, k$. 
The corresponding decision problem is called the {\sc Anchored Induced Topological Minor} problem.

\medskip
\noindent
{\bf Mutually Induced Paths.}
In the remainder of our paper we consider a slight generalization of the standard definition given in Section~\ref{s-intro}.
We say that paths $P_1,\ldots, P_k$ in a graph $G=(V,E)$  are {\it mutually induced} if 

\begin{itemize}
\item [(i)] each $P_i$ is an induced path in $G$;
\item [(ii)]  any distinct $P_i,P_j$ may only share vertices that are ends of both paths;
\item [(iii)]  no inner vertex $u$ of any $P_i$ is adjacent to a vertex $v$ of some $P_j$ for $j\neq i$, except when $v$ is an end-vertex of both $P_i$ and $P_{j}$.
\end{itemize}
This more general definition is exactly what we need later when detecting induced topological minors. In particular, this definition allows terminal pairs $(s_1,t_1),\ldots, (s_k,t_k)$ to have the following two properties:

\begin{itemize}
\item [1)]  for all  $i<j$, a terminal of $(s_i,t_i)$ may be adjacent to a terminal of $(s_j,t_j)$;
\item [2)] for all $i<j$, it holds that $0\leq |\{s_i,t_i\}\cap \{s_j,t_j\} |\leq 1$.
\end{itemize}
Property 2 means that terminal pairs may overlap but not coincide, that is, the set of terminal pairs is not a multiset.
This suffices for our purposes regarding detecting induced paths or cycles through specified vertices and detecting so-called anchored induced topological minors, as we will explain later.

For the remainder of the paper, we  assume that the input of 
{\sc Induced Disjoint Paths}
consists of a graph $G$ with a set of terminal pairs $(s_1,t_1),\ldots,(s_k,t_k)$ having Properties 1 and 2, and that the desired output is a set of paths $P_1,\ldots,P_k$ that are mutually induced, such that
$P_i$ has end-vertices $s_i$ and $t_i$ for $i=1,\ldots,k$. We say that $P_i$ is the {\it $s_it_i$-path} and also call it a  {\it solution path}.
We still call the subgraph of $G$ induced by the vertices of such paths an {\it induced linkage} and say that it forms a {\it solution for $G$}.

We observe that for general graphs, we can easily transform the variant with adjacent terminals or overlapping terminals pairs to the variant with neither adjacent terminals nor overlapping terminal pairs. First, adjacent terminals can be avoided by subdividing the edge between them. Second, a vertex $u$ representing $\ell\geq 2$ terminals can be replaced by $\ell$ new mutually non-adjacent vertices, each connected to all neighbors of $u$ via subdivided edges; even a situation with coinciding terminal pairs can be processed in this way. Third, let us recall that we may without loss of generality assume that every path $P_i$ is induced. However, these operations might not preserve claw-freeness, and hence we need different techniques in this paper (as we show later).

Our algorithm in Section~\ref{s-fpt} makes use of the aforementioned result of Robertson and Seymour on the {\sc $k$-Disjoint Paths} problem.

\begin{lemma}[\cite{RS95}]\label{t-RS}
For any fixed integer $k$, the {\sc $k$-Disjoint Paths} problem is solvable in $O(n^3)$ time for $n$-vertex graphs.
\end{lemma}

We also need the following terminology.
Let $G=(V,E)$ be a graph with terminal pairs $(s_1,t_1),\ldots, (s_k,t_k)$. By Property 2, a vertex $v$ can be a terminal in more than one terminal pair, e.g., $v=s_i$ and
$v=s_j$ is possible for some $i\neq j$. For clarity reasons, we will view $s_i$ and $s_j$ as two different terminals {\it placed on} vertex $v$.  
We then say that a vertex $u\in V$ {\it represents} terminal $s_i$ or $t_i$ if $u=s_i$ or $u=t_i$, respectively. We call such a vertex a {\it terminal vertex}; the other vertices of $G$ are called {\it non-terminal vertices}. We let $T_u$ denote the set of terminals represented by $u$ and observe that  $|T_u|\geq 2$ is possible.
We call two terminals that belong to the same terminal pair {\it partners}. 
We note that two partners may be represented by the same vertex, that is, $s_i$ and $t_i$ may belong to $T_u$ for some $u\in V$.

For our algorithm to work we first need to apply certain preprocessing operations. To this end, we introduce the following notation. Let $G$ be a  graph that together with terminal pairs $(s_1,t_1),\ldots,(s_k,t_k)$ forms an instance $I_1$ of the {\sc Induced Disjoint Paths} problem. We say that  an instance $I_2$ that consists of a graph $G'$ with terminal pairs $(s_1',t_1'),\ldots, (s_{k'}',t_{k'}')$  is {\it equivalent} to the first instance if the following three conditions hold:
\begin{itemize}
\item [(i)] $k'\leq k$;
\item [(ii)] $|V(G')|\leq |V(G)|$; 
\item [(iii)] $I_2$ is a  {\tt Yes}-instance if and only if $I_1$ is a {\tt Yes}-instance.
\end{itemize}
We say that an operation that transforms an instance of {\sc Induced Disjoint Paths} into a new instance {\it preserves the solution} if  the new instance is equivalent to the original instance. Most operations in our algorithm will be deletions of non-terminal vertices. In such cases preserving the solution just comes down to checking if the new instance still has a solution whenever the original instance has one.

In our algorithm, we sometimes have to solve the 
{\sc Induced Disjoint Paths}
problem on a graph that contain no terminals as a subproblem. We consider such instances
{\tt Yes}-instances (that have an empty solution). 

\medskip
\noindent
{\bf Graph Classes.}
The graph $K_{1,k}$ denotes the star with $k$ rays.
In particular, the graph $K_{1,3}=(\{a_1,a_2,a_3,b\},\{a_1b,a_2b,a_3b\})$ is called a {\it claw}. 
A graph is {\it $K_{1,k}$-free} if it has no induced
subgraph isomorphic to $K_{1,k}$. If $k=3$, then we usually call such a graph {\it claw-free}.

The {\it line graph} of a graph $G$ with edges $e_1,\ldots,e_p$ is the graph
$L(G)$ with vertices $u_1,\ldots,u_p$ such that there is an edge between any
two vertices $u_i$ and $u_j$ if and only if $e_i$ and $e_j$ share one end
vertex in $G$. Every line graph is claw-free.
We call $G$ the {\it preimage} of $L(G)$. 
It is well known that every connected line
graph except $K_3$ has a unique preimage (see \eg\cite{H}).

As a subroutine of our algorithm in Section~\ref{s-fpt}, we must compute the preimage of a line graph. For doing this we can use the linear-time algorithm of Roussopoulos~\cite{Ro73}.

\begin{lemma}[\cite{Ro73}]\label{l-preimage}
There exist an  $O(\max\{m,n\})$ algorithm for determining the preimage from a line graph $G$ on $n$ vertices and $m$ edges.
\end{lemma}

We also need the following lemma (see e.g. Ryj\'{a}\v{c}ek~\cite{Ry97}, who showed that a graph as in the lemma statement even has a triangle-free preimage).

\begin{lemma}\label{l-linegraph}
Every graph in which the neighborhood of every vertex induces a disjoint union of at most two cliques is a line graph.
\end{lemma}

A graph is an {\it interval graph} if intervals of the real line can be
associated with its vertices such that two vertices are adjacent if and only
if their corresponding intervals 
intersect. An interval graph is {\it proper} if it has an interval representation in which no interval is properly contained in any other interval.
Analogously, we can define the class of {\it circular-arc graphs} and {\it proper circular-arc graphs} by considering
a set of intervals 
(arcs) on the circle instead of a real line.
Proper interval graphs and proper circular-arc graphs are also known as linear interval graphs and circular interval graphs, respectively (cf.~\cite{FKLP12}). 
Note that proper circular-arc graphs are claw-free. 
We will use the following lemma due to Deng, Hell and Huang~\cite{DHH96}.

\begin{lemma}[\cite{DHH96}]\label{l-deng}
Proper circular-arc graphs can be recognized 
in linear time. A corresponding interval representation of such graphs can be constructed in linear time as well.
\end{lemma}

We also need the following result due to Golovach et al.~\cite{GPV13}.\footnote{It should be noted that our algorithm for claw-free graphs
will only apply 
Lemma~\ref{l-circ} for proper circular-arc graphs, and for those graphs we already showed that {\sc Induced Disjoint Paths} is linear-time
solvable in the extended abstract of this paper~\cite{GPV12b}.}

\begin{lemma}[\cite{GPV13}]\label{l-circ}
The {\sc Induced Disjoint Paths} problem can be solved in $O(n+m+k)$ time on circular-arc graphs on $n$ vertices and $m$
edges with $k$ terminal pairs.
\end{lemma}

\medskip
\noindent
{\bf Structure of Claw-Free Graphs.}
Chudnovsky and Seymour have given a structural characterization for claw-free graphs, the proof of which can be found in a series of seven papers called {\it Claw-free graphs I} through {\it VII}. We refer to their survey~\cite{CS05} for a summary.
Hermelin~\etal\cite{HMVW11} gave an algorithmic version of their result. 
This version plays an important role in the proof of our main result in Section~\ref{s-fpt}. In order to state it we need some additional terminology. 

Two adjacent vertices $u$ and $v$ in a graph $G$ are called {\it (true) twins} if they share the same neighbors, \ie 
$N[u] = N[v]$. 
The equivalence classes of the twin relation are called \emph{twin sets}.

The following result is implicit in a paper by Habib, Paul and Viennot~\cite{HPV1998}.

\begin{lemma}[\cite{HPV1998}] \label{l-twinsets}
The problem of detecting all twin sets in an $n$-vertex graph with $m$ edges is solvable in $O(n+m)$ time.
\end{lemma}

Two disjoint cliques $A$ and $B$ form a \emph{proper W-join} in a graph $G$ if 
$|A|\geq 2$, $|B|\geq 2$,
every vertex $v\in V(G)\setminus(A \cup B)$ is
either adjacent to all vertices of $A$ or to no vertex of $A$, every vertex in $A$ is adjacent to at least one vertex in $B$ and non-adjacent to at least one vertex in $B$, and the above also holds with $A$ and $B$ reversed.

We need the following result by King and Reed~\cite{KR08}.

\begin{lemma}[\cite{KR08}]\label{l-properwjoin}
The problem of detecting a  proper W-join in an $n$-vertex graph with $m$ edges is solvable in $O(n^2m)$ time.
\end{lemma}

A {\it hypergraph} is  a pair $R=(V_R,E_R)$ where $V_R$ is a set of elements called {\it vertices} and $E_R$ is a collection of subsets of $V_R$ called
{\it hyperedges}. Two hyperedges $e_1$ and $e_2$ are {\it parallel} if they contain the same vertices of $V_R$. Graphs can be seen as hypergraphs in which all 
hyperedges have size two.

A \emph{strip-structure} $(R, \{(J_{e},Z_{e}) \mid e \in E(R)\})$ for a claw-free graph $G$ is a hypergraph $R$, with possibly 
parallel and 
empty
hyperedges, and a set of tuples $(J_{e},Z_{e})$ for each $e \in E(R)$ called \emph{strips} such that
\begin{itemize}
\item [$\bullet$] $J_{e}$ is a claw-free graph and $Z_{e} \subseteq V(J_{e})$,
\item [$\bullet$] $\{V(J_{e}) \setminus Z_{e} \mid e \in E(R)\}$ is a partition of $V(G)$ and each $V(J_{e}) \setminus Z_{e}$ is nonempty,
\item [$\bullet$] $J_{e}[V(J_{e}) \setminus Z_{e}]$ equals $G[V(J_{e}) \setminus Z_{e}]$,
\item [$\bullet$] each $v \in e$ corresponds to a unique $z_{v} \in Z_{e}$ and vice versa.
\item [$\bullet$] for each $v \in V(R)$, the set $C_{v} := \bigcup_{z_{v} \in Z_{e} : v \in e} N_{J_{e}}(z_{v})$ induces a clique in $G$,
\item [$\bullet$] each edge of $G$ is either in $G[C_{v}]$ for some $v \in V(R)$ or in $J_{e}[V(J_{e}) \setminus Z_{e}]$ for some $e \in E(R)$.
\end{itemize}
Note that a vertex $v\in V(R)$ may be in more than one hyperedge of $R$, say $v$ belongs to $e_1,\ldots,e_r$ for some $r\geq  1$. Then $v$ corresponds to some unique vertex $z_v(e_i)$ in every $Z_{e_i}$, and the union of the neighbors of these $z$-vertices form a clique in $G$. 
When there is no confusion about $e$, we drop the subscript {\it e} and just talk about strips $(J,Z)$. 

A strip $(J,Z)$ is called a \emph{stripe} if the vertices of $Z$ are pairwise nonadjacent and any vertex in $V(J)\setminus Z$ is adjacent to at most one vertex of $Z$. A strip $(J,Z)$ is called a \emph{spot} if $J$ is a three-vertex path and $Z$ consists of both ends of this path.

We can now state the required lemma, which is easily derived from Lemma~C.20 in~\cite{HermelinMvLW-arxiv} or Theorem~1 in~\cite{HMVW11}.

\begin{lemma}[\cite{HMVW11,HermelinMvLW-arxiv}]\label{l-structure}
Let $G$ be a connected claw-free graph, such that G does not admit twins or proper W-joins and $\alpha(G) > 4$. Then either
\begin{itemize}
\item [1.] $G$ is a proper circular-arc graph, or 
\item [2.] $G$ admits a strip-structure such that each strip $(J,Z)$ either is
\begin{itemize}
\item [(a)] a spot, or
\item [(b)] a stripe with $|Z|=1$ and $J$ is proper circular-arc  or has $\alpha(J)\leq 3$, or
\item [(c)] a stripe with $|Z|=2$, and $J$ is proper interval  or has $\alpha(J)\leq 4$. 
\end{itemize}
\end{itemize}
Moreover, it is possible to distinguish the cases and to find the strip-structure in polynomial time.
\end{lemma}

\medskip
\noindent
{\bf Parameterized Complexity Theory.} In this theory, we consider the problem input as a pair $(I,k)$, where $I$ is the main part and $k$ the parameter. 
A problem $\Pi$ is \emph{fixed-parameter tractable} if an instance $(I,k)$ can be solved
in time $O(f(k)|I|^c)$, where $f$ denotes a computable function and $c$ a constant independent of $k$. 
A reduced instance $(I', k')$ of a problem $\Pi$ for an instance $(I, k)$ of $\Pi$ is called a
\emph{kernel} if the following three conditions hold:
\begin{itemize}
\item [(i)] $k'\leq k$ and $|I'|\leq g(k)$ for some computable function $g$;
\item [(ii)] the reduction from $(I, k)$ to $(I', k')$ is computable in polynomial time;
\item [(iii)] $(I,k)$ is a {\tt Yes}-instance of $\Pi$ if and only if $(I',k')$ is a {\tt Yes}-instance of $\Pi$.
\end{itemize}
The upper bound $g(k)$ of $|I'|$ is called the \emph{kernel size}.
 It is well known that a parameterized problem
is fixed-parameter tractable if and only if it is decidable and 
kernelizable (cf.~\cite{Niedermeier06}). In the latter case, the next step is to research whether the problem has a \emph{polynomial} kernel, that is, allows a kernel size which is polynomial in $k$.
We refer to the monographs of
Downey and Fellows~\cite{DF99} and Niedermeier~\cite{Niedermeier06} for more on the theory of parameterized complexity.

\section{Mutually Induced Disjoint Paths}\label{s-fpt}

In this section we present the following result.

 \begin{theorem}\label{t-main}
The {\sc Induced Disjoint Paths} problem is fixed-parameter tractable on claw-free graphs when parameterized by $k$.
\end{theorem}

Below, we outline the general approach of our algorithm, and then give the details of the subroutines that we use. First, we need some definitions. Let $G$ be a graph with terminal pairs $(s_1,t_1),\ldots (s_k,t_k)$ that forms 
an instance of {\sc Induced Disjoint Paths}. 
We call this instance {\it claw-free}, {\it twin-free}, or {\it proper W-join-free} if $G$ is 
claw-free, twin-free, or proper W-join-free, respectively. In addition, we call the instance {\it independent}  
if the terminal vertices form an independent set, and no terminal vertex represents two terminals of the same terminal pair. Note that in this definition it is still possible for a terminal vertex to represent more than one terminal.
However, no terminal vertex in an independent instance can represent more than two terminals if the instance has a solution and $G$ is claw-free. Otherwise, any solution would induce a claw in the neighborhood of this terminal vertex.

In our algorithm, it may happen that the graph under consideration gets disconnected. In that case we make the following implicit check. We stop considering this graph if there is a terminal pair of which the terminals are in two different connected components. Otherwise we consider each connected component separately. Hence, we may assume that the graph is connected.

\medskip
\noindent \rule{\textwidth}{0.1mm}\\
{\sc The Algorithm and Proof of Theorem~\ref{t-main}}\label{a-algorithm} 

\medskip
\noindent
Let a claw-free graph $G$ on $n$ vertices with terminal pairs $(s_1,t_1),\ldots, (s_k,t_k)$ for some $k\geq 1$ form an instance. 

\medskip
\noindent
{\bf Step 1. Reduce to an independent instance.}\\
We apply Lemma~\ref{l-pre} 
(stated later)
and obtain in $O(k^{2}n)$ time an independent and  equivalent instance that consists of a claw-free graph with at most $k$ terminal pairs.
For simplicity, we denote this graph and these terminals pairs by $G$ and $(s_1,t_1)$,$\ldots$, $(s_k,t_k)$ as well.

\medskip
\noindent
{\bf Step 2. Solve the problem if $\alpha$ is small.}\\
Because all terminal vertices are independent, we find that 
$k\leq \alpha(G)$ 
holds. Hence, if $\alpha(G)\leq 4$, we can solve the problem
by applying the aforementioned $n^{O(k)}$ time algorithm of Fiala~\etal\cite{FKLP12}. From now on we assume that $\alpha(G)>4$. 

\medskip
\noindent
{\bf Step 3. Remove twins.}\\
We apply Lemma~\ref{l-twin} 
and obtain in linear time an independent and  equivalent instance that consists of a claw-free,  twin-free graph with the same $k$ terminal pairs as before.
For simplicity, we denote the new graph by $G$ as well.

\medskip
\noindent
{\bf Step 4. Remove proper W-joins.}\\
We apply Lemma~\ref{l-properwjoin2} and obtain in $O(n^5)$ time an 
independent and  equivalent instance that consists of a claw-free, twin-free, proper W-join-free graph with the same $k$ terminal pairs as before.
For simplicity, we denote the new graph by $G$ as well.

\medskip
\noindent
{\bf Step 5. Solve the problem for a proper circular-arc graph.}\\
By Lemma~\ref{l-deng} we can check in linear time if $G$ is a proper circular-arc graph. If so, then
we apply Lemma~\ref{l-circ} to solve the problem in linear time.
From now on we assume that $G$ is not a circular-arc graph, and hence not a proper circular-arc graph.

\medskip
\noindent
{\bf Step 6.  Reduce to a collection of line graphs.}\\
By Lemma~\ref{l-structure} we find in polynomial time a strip-structure of $G$, in which each strip $(J,Z)$ is either
a spot, or a stripe with $|Z|=1$, and $J$ is proper circular-arc  or has $\alpha(J)\leq 3$, or
a stripe with $|Z|=2$, and $J$ is proper interval  or has $\alpha(J)\leq 4$. 
We apply Lemma~\ref{l-strip}, and in $6^kn^{O(1)}$ time either find that the instance has no solution, or obtain at most  $6^k$ line graphs on at most $n$ vertices and with at most $k$ terminals each, such that $G$ has a solution if and only if at least one of these
line graphs has a solution. 

\medskip
\noindent
{\bf Step 7. Solve the problem for each line graph.}\\
For each of the $6^k$ line graphs $G'$ created we can do this in $O(g(k)|{V_G'}|^6)$ time due to Lemma~\ref{l-line}.
Here, $g(k)$ is a function that only depends on $k$.
We conclude that our algorithm runs in 
$6^kg(k)n^{O(1)}$ time, as desired.

\smallskip
\noindent \rule{\textwidth}{0.1mm}

\medskip
\noindent
To finish the correctness proof and running time analysis of our algorithm, 
it remains to state and prove the missing lemmas, namely Lemmas~\ref{l-pre}--\ref{l-line}.
We do this as follows.
In Section~\ref{s-preprocessing} we show Lemmas~\ref{l-pre}--\ref{l-properwjoin2}, which are related to the preprocessing
of the input graph, that is, which are used in Steps 1--4. 
In Section~\ref{s-strips} we show Lemma~\ref{l-strip}, which we used in Step~6 where we obtain a set of line graphs. 
Finally, in Section~\ref{s-linegraphs} we show Lemma~\ref{l-line}, which we used in Step~7 where we solve the problem for line graphs.

\subsection{Independent, Twin-Free, and Proper W-Join-Free Instances}\label{s-preprocessing}

In this section we state and prove Lemmas~\ref{l-pre}--\ref{l-properwjoin2} used in Steps~1,3,4, respectively.

\begin{lemma}\label{l-pre}
There is a $O(k^{2}n)$-time  algorithm that transforms an instance consisting of an $n$-vertex
claw-free graph $G$ with $k$ terminal pairs
into an equivalent instance that is independent and claw-free.
\end{lemma}

\begin{proof}
Let $G=(V,E)$ be a claw-free graph that has terminal pairs $(s_1,t_1)$,$\ldots$, $(s_k,t_k)$ for some $k\geq 1$. 
Let $|V|=n$, $|E|=m$, and let $T$ denote the set of all terminals, that is, $T=\bigcup_{u\in V}T_u$. 

We apply a number of operations on $G$ in order to obtain a new instance that consists of a graph $G'$ with terminal pairs $(s_1',t_1'),\ldots, (s_{k'}',t_{k'}')$. 
These operations are of three different types: the first type of operation removes one or more non-terminal vertices from $G$, the second type of operation only removes one or more terminals with their partners from $T$ without modifying $G$, and the third type of operation removes edges from $G$.
Hence, we will have $k'\leq k$ and $|V(G')|\leq |V(G)|$. We will prove that the new instance is a {\tt Yes}-instance if and only if the original instance is by showing that our operations preserve solutions and claw-freeness. Afterwards, we will see that the new instance is independent, and we will analyze the overall running time.

The following four rules describe the operations in detail. They must be applied consecutively, starting with Rule 1 and ending with Rule 4.

\medskip 
\noindent
{\bf Rule 1.}
Remove every non-terminal vertex $u$ from $G$ that is adjacent to two adjacent terminal vertices $v$ and $w$.

\medskip
\noindent
Such an operation preserves the solution as the removed vertices cannot be part of any solution path for $G$.
This can be seen as follows. 
If $v$ and $w$ represent terminals from different pairs, then $u$ cannot be used as (inner) vertex for any solution path.
Suppose that $v$ and $w$ each represent a terminal of the same pair, say $s_i\in T_v$ and $t_i\in T_w$. 
Then we may assume without loss of generality that $P_i=vw$.
Since the set of terminal pairs is not a multiset, $u$ cannot be used as an (inner) vertex of some other solution
path $P_j$.
Because we only removed vertices from $G$ and claw-free graphs are closed under vertex deletion, the resulting graph remains claw-free. 

\medskip
\noindent
{\bf Rule 2.}
Find the set $U$ of all terminal vertices $u$ such that $u$ only represents terminals  whose partners are in $N_G[u]$.
Remove $U$ and all non-terminal vertices of $N_G(U)$ from $G$. Remove from $T$ the terminals of all terminal pairs $(s_i,t_i)$ such that $s_i\in T_u$ or $t_i\in T_u$
for $u\in U$.

\medskip
\noindent
By the same argument that we used for Rule 1, the resulting graph is claw-free.

To show that this rule preserves the solution, we prove first that if the original instance has a solution, then the new instance has a solution. Let $P_i$ be an $s_it_i$-path in a solution for the original instance with $\{s_i,t_i\}\cap T_u=\emptyset$ for all $u\in U$. Then $P_i$ contains no internal vertices in $N_G[U]$. Observe also that the terminal vertices representing $s_i$ and $t_i$ are not deleted by Rule~2. This means that $P_i$ is a path in the graph for the new instance. Consequently, the paths $P_i$ for those terminal pairs $(s_i,t_i)$ with
$\{s_i,t_i\}\cap T_u=\emptyset$ for all $u\in U$ compose a solution for the new instance. 

Assume now that the new instance has a solution. Let $P_i$ be an $s_it_i$-path in the solution, and let $x,y$ be the vertices that represent $s_i,t_i$ respectively. Then $\{s_i,t_i\}\cap T_u=\emptyset$ for all $u\in U$. 
Let $v$ be an internal vertex of $P_i$. 
We show that if $v$ is adjacent to a vertex $z$ that represents a terminal in the original instance, then $z=x$ or $z=y$.  
Suppose that $z\notin \{x,y\}$.
Since $P_i$ is a path in a solution for the new instance, $z$ does not represent a terminal in the new instance. Therefore, $z$ may only represent terminals that were deleted from $T$ using Rule~2. Then $z$ is in $U$, because if it would only represent the partner of terminals that are in $U$, then by the definition of $U$ these partners are all in $N_{G}(z)$ and thus $z$ must belong to $U$ as well. However, as $P_{i}$ cannot have any internal vertices in $N_{G}[U]$ by Rule~2, $z$ cannot belong to $U$, a contradiction.

Now we can construct the solution for the original instance from the solution for the new instance by adding new paths as follows. 
For any terminal pair $(s_j,t_j)$ such that $s_j,t_j$ are represented by $u,v$ respectively and $u\in U$ or $v\in U$, let $P_j=u$ if $u=v$ and let $P_j=uv$ if $u\neq v$, and add $P_j$ to the solution. 
Then, because all paths from the solution to the new instance avoid the neighborhood of terminals that were deleted by Rule~2 (as argued above), the constructed set of paths is indeed a solution to the original instance.

\medskip
\noindent
{\bf Rule 3.}
Remove the terminals of every terminal pair $(s_i,t_i)$ with $s_i\in T_u$ for some $u\in V(G)$ and $t_i\in T_v$ for some $v\in N_G[u]$ from $T$.

\medskip
\noindent
This preserves the solution because of the following reasons. If the original instance has a solution, then the new instance has a solution. 
Now suppose that the new instance has a solution. We extend this solution by adding the path $P_i=u$ if $v=u$ or $P_i=uv$ otherwise.
Because we already applied Rule~2, we find that 
$T_u$ must contain at least one terminal whose partner is not represented by a vertex in $N_G[u]$.
In that case, $u$ is only adjacent to inner vertices of other solution paths that start in $u$.
Similarly, if $v\neq u$, then $v$ is only adjacent to inner vertices of other solution paths that start in $v$.
Hence, we may extend the solution of the new instance by the path $P_i$ in order to obtain a solution of the original instance.
The resulting graph is claw-free, as we did not touch $G$ and only modified $T$.

\medskip
\noindent
{\bf Rule 4.}
For every pair of adjacent terminal vertices, remove the edge between them.

\medskip
\noindent
This preserves the solution, because in this stage two adjacent terminal vertices do not represent two terminals from the same pair; otherwise we would have applied Rule~3 already. Now suppose that the resulting graph contains a claw with center $u$ and leaves $v,w,x$. Because we preserved claw-freeness so far, there must have been an edge between two leaves, say between $v$ and $w$. This implies that $v$ and $w$ are terminal vertices. 
Then $u$ must be a non-terminal vertex, as otherwise we would have removed the edges $uv$ and $uw$ as well. However, this is not possible, because we would have removed $u$ when applying Rule 1. Hence, the resulting graph is claw-free.

\medskip
\noindent
Already after applying Rule 3, there is no terminal vertex that represents two terminals from the same terminal pair. After applying Rule 4, all terminal vertices are independent.

It remains to analyze the running time. 
We assume that the graph is given as an adjacency list and an adjacency matrix. 
Rule~1 can be implemented in $O(k^{2}n)$ time: for each pair of adjacent terminal vertices, we check all non-terminal vertices for being a common neighbor.  
Rule~2 can be implemented in $O(k^2+kn)$ time: we first find the set $U$ (which has size at most $k$) in time $k^2$ and then 
delete $U$ together with the non-terminal vertices of $N_G(U)$ in time $O(kn)$ and modify $T$ in time $O(k)$.
Rule~3 can be implemented in $O(k)$ time by testing whether the terminals in each terminal pair are represented by the same vertex or by two adjacent vertices.
Rule 4 can be implemented in $O(k^2)$ time 
by testing whether any two terminals from different terminal pairs are represented by
adjacent vertices. 
We conclude that the total running time is $O(k^{2}n)$.
This completes the proof of Lemma~\ref{l-pre}.\qed
\end{proof}

\begin{lemma}\label{l-twin}
There is a linear-time algorithm that transforms an independent instance consisting of an $n$-vertex, claw-free graph $G$ with $k$ terminal pairs
into an equivalent instance that is independent, claw-free, and twin-free.
\end{lemma}
\begin{proof}
Let $G=(V,E)$ be a claw-free graph on $n$ vertices with $k$ terminal pairs $(s_1,t_1),\ldots, (s_k,t_k)$ for some $k\geq 1$ that forms an independent instance. We first find all twin sets of $G$ in linear time using Lemma~\ref{l-twinsets}. Let $A$ be a twin set of size at least two. 
Because the terminal vertices form an independent set, at most one vertex of $A$ is a terminal vertex. If one vertex of $A$, say $u$, is a terminal vertex, then we remove 
$A\setminus\{u\}$ from $G$. In the other case, \ie if $A$ does not contain a terminal vertex, we arbitrarily choose a vertex $v$ from $A$ and remove $A \setminus\{v\}$ from $G$. 
In both cases we preserve the solution, because all removed vertices are non-terminal vertices that cannot occur as an (inner) vertex in a solution path of a solution for $G$. We let $G'$ denote the twin-free graph obtained after applying this operation as long as possible.

Because we only removed non-terminal vertices, we find that $G'$ has the same set of $k$ terminals, which still form an independent set. Moreover, $|V(G')|\leq |V(G)|$ holds. Because $G$ is claw-free and the class of claw-free graphs is closed under vertex deletion, we find that $G'$ is claw-free. 
Clearly, the above procedure runs in linear time.
This completes the proof of Lemma~\ref{l-twin}.\qed
\end{proof}

\begin{lemma}\label{l-properwjoin2}
There is an $O(n^5)$-time algorithm that transforms an independent instance consisting of an $n$-vertex, claw-free, twin-free graph $G$ with $k$ terminal pairs
into an equivalent instance that is independent, claw-free, twin-free, and proper W-join-free.
\end{lemma}
\begin{proof}
Let $G=(V,E)$ be a claw-free, twin-free graph on $n$ vertices with $k$ terminal pairs $(s_1,t_1),\ldots, (s_k,t_k)$ for some $k\geq 1$ that forms an independent instance.
Because of the latter property, 
every clique in every proper W-join in $G$ contains at most one terminal vertex.
This means that we only have to distinguish between the following four cases for every proper W-join $(A,B)$.
In this case analysis, we also  assume that every terminal vertex represents at most two terminals, because our instance is independent and $G$ is claw-free.

\medskip
\noindent
{\bf Case 1.} At least one of $A,B$, say $A$, contains a terminal vertex $u$ that represents exactly two terminals.

\medskip
\noindent
We remove all vertices of $A\setminus \{u\}$ from $G$.
This preserves the solution, which can be seen as follows.
Let $(s_{i},t_{i})$ and $(s_{j},t_{j})$ be the two terminal pairs with a terminal represented by $u$. Without loss of generality, $u$ represents $s_{i}$ and $s_{j}$.
Let $v\neq u$ be a vertex of $A$.
Suppose that we have a solution for $G$ that has a solution path $P$ containing $v$.  By the definition of a solution, $P$ must have $u$ as one of its two end-vertices. Then we may assume without loss of generality that $P$ is the 
$s_it_i$-path. Let $w\neq u$ be the other neighbor of $v$ on $P$. This neighbor exists,
because $v$ is not a terminal vertex. Since every vertex not in $A\cup B$ is either adjacent to all vertices of $A$ or to none of them, $w$ must be in $B$; otherwise $P$ is not induced.

Now consider the solution path $P'$ of this solution that connects terminals $s_j$ and $t_j$. 
Because $u$ represents $s_j$, we find that $P'$ also starts in $u$. 
Since $s_j$ and $t_j$ are represented by two different non-adjacent vertices, we find that the neighbor of $u$ on $P'$ is an inner vertex of $P'$. Let $v'$ be this neighbor.  If $v'\in V\setminus (A\cup B)$, then $v'$ is adjacent to $v$ by the definition of a proper $W$-join. This is not possible, because inner vertices of two different solution paths are not adjacent. For exactly the same reason we find that $v'\notin A$. Hence, $v'\in B$.  But then $v'$ and $w$ are adjacent. This is not possible either by the definition of a solution. We conclude that removing all vertices in $A\setminus \{u\}$ preserves the solution.
As $A$ now has size 1, $(A,B)$ is no longer a proper W-join.

\medskip
\noindent
{\bf Case 2.}  The cliques $A$ and $B$ each have exactly one terminal vertex $u$ and $v$,  respectively, that each represent exactly one terminal.

\medskip
\noindent
First suppose that the two terminals represented by $u$ and $v$ are from different terminal pairs. We assume without loss of generality that $u$ represents terminal $s_i$ and that $v$ represents terminal $s_j$, where $i\neq j$. 
We remove all vertices of $A\setminus \{u\}$ from $G$.
We claim that this preserves the solution.
In order to obtain a contradiction, assume that we have a solution for $G$ that has a solution path $P$ containing a vertex $w\in A\setminus \{u\}$.
Because $u$ and $w$ both belong to $A$, they are adjacent. Hence, $P$ must be the $s_it_i$-path. 
Because terminal vertices are non-adjacent and $u$ is a terminal vertex, $w$ is not a terminal vertex. This means that $w$ has a neighbor $w'\neq u$ on $P$. 
If $w'\in A$, then $P$ is not induced. Also if $w'\in V\setminus (A\cup B)$, then $P$ is not induced; this follows from the definition of a proper W-join. Hence, we find that
$w'\in B$. Because $v$ is the only terminal vertex in $B$, this means that $w'$ is not a terminal vertex. Because $v$ and $w$ both belong to $B$, they are adjacent.
However, $v$ only represents $s_j$ and $i\neq j$. Hence, we obtain a contradiction. 

Now suppose that the two terminals represented by $u$ and $v$ are from the same terminal pair, and say $u$ represents terminal $s_i$ and $v$ represents terminal $t_i$. Let $w$ be a neighbor of $u$ in $B$; 
note that $w\neq v$ as our instance is independent. 
We remove $N[A \cup B]$ and the terminal pair $(s_{i},t_{i})$. We claim that this preserves the solution. This can be seen as follows. Suppose that we have a solution for $G$. 
Let $P$ be the $s_it_i$-path. 
Because $(A,B)$ is a proper W-join, $N[u]\cup N[v]=N[A\cup B]$.
Since $u$ and $v$ only represent $s_i$ and $t_i$, respectively, and are the only terminal vertices in $A\cup B$, the only solution path that can use a vertex from 
$N[u]\cup N[v]=N[A\cup B]$ is $P$.
Consequently, removing $P$ results in a solution for the resulting instance.
Moreover, if we have a solution for the resulting instance, we extend
it to a solution for $G$ by adding the $s_it_i$-path $uwv$.

\medskip
\noindent
{\bf Case 3.} Exactly one of $A,B$, say $A$, contains a terminal vertex, and this terminal vertex represents exactly one terminal.

\medskip
\noindent
Let $u\in A$ be this terminal vertex. Let $s_i$ be the terminal represented by $u$. We remove all vertices of $A\setminus \{u\}$ from $G$. We claim that this preserves the solution. This can be seen as follows. Suppose that we have a solution for $G$. Let $P$ be the $s_it_i$-path. Because $u$ is the only terminal vertex and $u$ represents only one terminal, the only solution path that uses a vertex from $A\setminus \{u\}$ is $P$. Let $v$ be the neighbor of $u$ on $P$. If $v\notin A$, then we are done. Suppose that $v\in A$. Because $A$ only contains $u$ as a terminal vertex, $v$ is an inner vertex of $P$. Consequently, $v$ has another neighbor on $P$ besides $u$. Let $w$ be this neighbor. Because $u$ and $v$ are in $A$, we find that $w\notin N(A)\setminus B$. Hence $w\in B$. Then we reroute $P$ by replacing $v$ and $w$ by a neighbor of $u$ in $B$; such a neighbor exists by the definition of a proper W-join.

\medskip
\noindent
{\bf Case 4.} Neither $A$ nor $B$ contains a terminal vertex.

\medskip
\noindent
By definition, $A$ contains two vertices $u$ and $v$ such that $u$ has a neighbor $w\in B$ that is not adjacent to $v$.
We remove all vertices of $(A\cup B)\setminus \{u,v,w\}$. We claim that this preserves the solution. This can be seen as follows.
Suppose that we have a solution for $G$. 
Then at most one vertex of $A$ and at most one vertex of $B$ is used as an (inner) vertex of some solution path; otherwise we would have a solution path that is not induced, because $(A,B)$ is a proper W-join.
If no solution path uses an edge between a vertex from $A$ and a vertex from $B$,  then we can reroute solution paths by replacing a vertex in $A$ by $v$ and a solution 
vertex in $B$ by $w$ if necessary. In the other case, if there is a solution path that uses such an edge, then we can reroute this solution path by replacing the end-vertices of this edge by $u$ and $w$ if necessary.

\medskip
In each of the four cases, we destroy the proper W-join. Note that we do this by removing one or more non-terminal vertices from $G$, 
and in addition by removing two terminal vertices representing terminals of a terminal pair if the second subcase of Case 2 occurred.
As such, the resulting graph has fewer vertices than $G$ and together with the remaining terminal pairs $(s_1,t_1),\ldots,(s_k',t_k')$ forms an independent and claw-free instance.

We may have created new twins. However, Lemma~\ref{l-twin} tells us that we can make the resulting graph twin-free, while preserving all the other properties. Hence, applying the two rules ``destroy a proper W-join'' and ``make the graph twin-free" consecutively and as long as possible yields an equivalent instance that is independent, claw-free, twin-free, and proper W-join-free.

We are left to analyze the running time. 
We can find a proper W-join in $O(n^4)$ time by Lemma~\ref{l-properwjoin}. 
Distinguishing the right case and applying the corresponding rule takes $O(n)$ time. Afterwards, we have removed at least one vertex.
Every call to Lemma~\ref{l-twin} takes $O(n^2$) time.  
We conclude that the total running time is $O(n(n^4+n+n^2))=O(n^5)$. This completes the proof of Lemma~\ref{l-properwjoin2}. \qed
\end{proof}

\subsection{Strips, Spots, and Stripes}\label{s-strips}

In this section we state and prove Lemma~\ref{l-strip} used in Step 6.

\begin{lemma}\label{l-strip}
Let $G$ be a graph that together with a set $S$ of $k$ terminal pairs forms a claw-free, independent, and twin-free instance of 
{\sc Induced Disjoint Paths}.
Let $(R,\{(J_e,Z_e) \mid e\in E(R)\})$ be a strip-structure for $G$, in which each strip $(J,Z)$ either is
\begin{itemize}
\item[1.] a spot, or
\item [2.] a stripe with $|Z|=1$, and $J$ is proper circular-arc  or has $\alpha(J)\leq 3$, or
\item [3.] a stripe with $|Z|=2$, and $J$ is proper interval  or has $\alpha(J)\leq 4$. 
\end{itemize}
There is a 
$6^kn^{O(1)}$-time 
algorithm that either shows that $(G,S)$ has no solution, or 
produces a set ${\cal G}$ of at most $6^k$ graphs,
such that
each $G'\in {\cal G}$ is a line graph with at most $|V(G)|$ vertices and at most $k$ terminal pairs, and such that 
$G$ has a solution if only if at least one graph in ${\cal G}$ has a solution.
\end{lemma}

\begin{proof}
Let $G$ be an $n$-vertex graph with a set of $k$ terminal pairs that has the properties as described in the statement of the lemma. 

Our algorithm is a branching algorithm 
that 
applies a sequence of graph modifications to $G$ until a line graph remains in each leaf of the branching tree. While branching, the algorithm keeps the terminal set and the strip structure up to date with the modifications being performed. This is possible dynamically, \ie without needing to recompute a strip structure from scratch, and no new strips are created in the algorithm. Moreover, the modifications ensure that all intermediate instances are claw-free and independent, \ie it is not necessary to reapply Lemma~\ref{l-pre}. Finally, we note that the modifications may remove some or all of the vertices of a strip. For example, for a strip $(J,Z)$, it may be that we remove $N[z]$ for some $z \in Z$, thus reducing the size of $Z$. Hence, at any time during the algorithm, a strip $(J,Z)$ is either
\begin{itemize}
\item[1.] a spot, or
\item [2.] a stripe with $|Z| = 1$, and $J$ is proper circular-arc  or has $\alpha(J)\leq 3$, or
\item [3.] a stripe with $|Z|=2$, and $J$ is proper interval  or has $\alpha(J)\leq 4$. 
\end{itemize}
Observe that, for example, the deletion of $N[z]$ for some $z \in Z$ preserves membership of one of these categories. It is also worth noting that such a deletion may create twins in an intermediate instance. However, the algorithm only relies on the original instance being twin-free, and hence this poses no problem.

The algorithm considers each strip at most once in any path of the branching tree. The branching strategy that the algorithm follows for a strip $(J,Z)$ depends on a complex case analysis. The main distinction is between the case $|Z|=1$ and the case $|Z|=2$. As we shall see, we do not have to branch in the first case. However, in the second case we may have to do so. After processing a strip and possibly branching, we obtain
for each branch
a new, intermediate instance of the problem that consists of the induced subgraph $G'$ of remaining vertices of $G$ together with those terminal pairs of $G$ that are represented by terminal vertices in $G'$. We call this \emph{reducing to $G'$}. Then the algorithm considers the next strip of the updated strip structure. This strip is arbitrarily chosen from the set of remaining unprocessed strips.

Before we begin, we first recall a number of properties that we will use throughout the case analysis and prove one additional claim. We recall that $T_u$ denotes the set of terminals represented by $u$ and that no two partners are represented by $u$, as $G$ and its set of terminal pairs form an independent instance. The definition of being an independent instance also means that the set of terminal vertices is independent. The latter property together with the claw-freeness of $G$ implies that every terminal vertex represents at most two different terminals.

In the claim below, $J'$ denotes a (not necessarily proper) induced subgraph of $J$ and $S_{J'}$ denotes a set of 
at most $k$
terminal pairs in $J'$, which is not necessarily a subset of $S$.

\medskip
\noindent
{\it Claim 1. We can decide in $n^{O(1)}$ time whether an instance $(J',S_{J'})$ is a {\tt Yes}-instance.}

\medskip
\noindent
We prove Claim 1 as follows. Either $J'$ is a proper circular-arc graph (or even a proper interval graph) or $\alpha(J')\leq 4$ (or even $\alpha(J')\leq 3$).
In the first case, we use  Lemma~\ref{l-circ}. In the second case, we deduce that 
$k\leq \alpha(J')\leq 4$, and we can use the
$n^{O(k)}$ time algorithm of Fiala et al.~\cite{FKLP12} for solving {\sc Induced Disjoint Paths}. This proves Claim~1.

\medskip
\noindent
We are now ready to start our case analysis.
In this analysis, we sometimes write that we {\it solve the problem on} an induced subgraph $G'$ of $G$.
Then we implicitly assume that we solve the {\sc Induced Disjoint Paths} problem on $G'$, where
$G'$ has inherited those terminal pairs of $G$ that are represented by terminal vertices in $G'$. 

\medskip
\noindent
{\bf Case 1.} $|Z|=1$.\\
We write $H=G[J\setminus Z]$ and $F=G-H$.
Assume that $Z=Z_{e_1}$ with $e_1=\{v\}$. Let $e_2,\ldots,e_p$ be the other hyperedges of $R$ that contain $v$.
For $i=1,\ldots,p$, we let $z_v(e_i)$ denote the vertex in $Z_{e_i}$ corresponding to $v$. Let $X=N_{J_{e_1}}(z_v(e_1))$ and
$Y=N_{J_{e_2}}(z_v(e_2))\cup \cdots \cup N_{J_{e_p}}(z_v(e_p))$. By definition, $X$ and $Y$ are both nonempty, $X\cap Y=\emptyset$ and $X\cup Y$ is a clique in $G$. Moreover,
$Y$ separates $V(H)$ from $V(F)\setminus Y$ if $V(F)\setminus Y$ is non-empty. 

If $H$ contains no terminal vertices, then we remove all vertices of $H$ from $G$. We may do this, because no path in a solution for $G$ will use a
vertex from $H$. The reason is that such a path will need to pass through $Y$ at least twice. This is not possible, because $Y$ is a clique.
From now on we assume that $H$ contains at least one terminal vertex.

Below we split Case 1 in a number of subcases. In these subcases we solve the problem for $H$ or the graph obtained from $H$ by adding a new vertex
adjacent to every vertex in $X$. The latter graph is isomorphic to $J$, whereas $H$ is an induced subgraph of $J$. Hence, this subroutine takes $n^{O(1)}$ time due to Claim 1.

\medskip
\noindent
{\bf Case 1a.} $X$ contains at least one terminal vertex.\\
Because $X$ is a clique, $X$ contains exactly one terminal vertex. Let $u$ be this terminal vertex.

Suppose that there is a pair $(s_i,t_i)$ with $s_i\in V(H)\setminus X$ and $t_i\in F\setminus Y$. Then $\{s_i,t_i\}\cap T_u=\emptyset$.
We conclude that $G$ has no solution, because the $s_it_i$-path of 
any solution for $G$ must pass $X$ and as such contain a neighbor of the terminal vertex $u$ as one of its inner vertices.
This is not allowed as $u$ is an end-vertex of at least one other solution path. 
From now on, suppose that no such pair exists.

First suppose that all partners of the terminals in $T_u$ belong to $H-N_H[u]$. Then no path in any solution for $G$ will use a vertex from $Y$. Hence, we first solve
the problem for $H$. 
If the answer is {\tt No}, then $G$ has no solution. Otherwise, we reduce to $F-Y$.
 
Now suppose that all partners of the terminals in $T_u$ belong to $F-Y$.  If $u$ represents more than one terminal, then we return {\tt No}. The reason is that $u$ then represents two terminals from different terminal pairs. The corresponding paths in any solution for $G$ must both contain a vertex from $Y$. This is not possible, because $Y$ is a clique. Hence $u$ represents exactly one terminal. Then no path in any solution for $G$ will use a vertex from $N_H[u]$. Hence, we can first solve 
the problem for $H-N_H[u]$. If $H$ has no solution, then we return {\tt No}. 
Otherwise, we reduce to $F+u$.

Finally, in the remaining case, we may assume without loss of generality that $u$ represents two terminals $s_i$ and $s_j$, such that $t_i\in V(H)\setminus N_H[u]$ and
$t_j\in V(F)\setminus Y$. This means that we can first solve the problem for $H-(X\setminus \{u\})$.  
If $H-(X\setminus \{u\})$ has no solution, then we return {\tt No}. Otherwise, we reduce to $F+u$.

\medskip
\noindent
{\bf Case 1b.} $X$ contains no terminal vertices.\\
First suppose that there is no terminal pair that is {\it mixed}, i.e., has one of its terminals in $H-X$ and the other one in $F$. Then
we first solve the problem for $H$ and $H-X$. If neither $H$ nor $H-X$ has a solution, then we return {\tt No}.
If $H-X$ has a solution, then we reduce to $F$. If $H-X$ has no solution but $H$ has, then there is a solution path in every solution for $G$ that
uses a vertex from $X$. Hence, in that case, we reduce to $F-Y$. 

Now suppose that there is exactly one mixed terminal pair. Let $(s_i,t_i)$ be this pair, where we assume that $s_i\in H-X$ and $t_i\in F$. 
Let $v'$ denote a new vertex added to $H$ by making it adjacent to every vertex in $X$. Let $H^*$ denote the resulting graph. Assume that $v'$ represents exactly one terminal, which is a new terminal $t_i'$ that replaces the partner $t_i$ of $s_i$.
We first solve the problem for $H^*$; note that $H^*$ is isomorphic to $J$. 
If $H^*$ has no solution, then we return {\tt No}.

Otherwise we reduce to the graph $F^*$ that is obtained from $F$ by adding a new vertex $u'$ and a new vertex adjacent to all vertices of $Y$ and to $u'$, and letting $u'$ represent a new terminal $s_i'$ that is the new partner of $t_i$. 
Note that the above modification of $F$ into $F^*$ ensures that the resulting instance is independent.

Finally, suppose that there are two or more mixed terminal pairs. Then we return {\tt No}. The reason is that in that case every solution must contain at least two different paths that use a vertex from $X$. This is not possible, because $X$ is a clique.

\medskip
\noindent
{\bf Case 2.} $|Z|=2$.\\
If $(J,Z)$ is a spot, we do nothing. Hence we assume that $(J,Z)$ is a stripe. 
We write $H=G[J\setminus Z]$ and $F=G-H$.
Assume that $Z=Z_{e_1}$ with $e_1=\{v_1,v_2\}$. Let $e^h_2,\ldots,e^h_{p_h}$ be the other hyperedges of $R$ that contain $v_h$ for $h=1,2$.
For $h=1,2$ and $i=1,\ldots,p_h$, we let $z_v(e^h_i)$ denote the vertex in $Z_{e^h_i}$ corresponding to $v_h$. For $h=1,2$, let $X_h=N_{J_{e^h_1}}(z_v(e^h_1))$ and
$Y_h=N_{J_{e^h_2}}(z_v(e^h_2))\cup \cdots \cup N_{J_{e^h_p}}(z_v(e^h_{p_h}))$. 
Because $(J,Z)$ is a stripe, $X_1\cap X_2=\emptyset$.
Also by definition, we have that for $h=1,2$, the sets $X_h$ and $Y_h$ are both nonempty, $(X_1\cup X_2)\cap (Y_1\cup Y_2)=\emptyset$, and $X_h\cup Y_h$ is a clique in $G$. Moreover,
$Y_1\cup Y_2$ separates $V(H)$ from $V(F)\setminus (Y_1\cup Y_2)$, should $V(F)\setminus (Y_1\cup Y_2)$ be nonempty. 
As an aside, we note that $Y_1$ and $Y_2$ may share some vertex. In that case, such a vertex corresponds to a spot. Because $G$ is twin-free, there can be at most one such vertex.

If $H$ contains no terminal vertices, then we remove all vertices of $H$ from $G$ 
except the vertices on a shortest path from a vertex $u_1\in X_1$ 
to a vertex $u_2\in X_2$.
We may do this, because every path $P$ in any solution for $G$ cannot use just one vertex from $H$; in that case such a vertex will be in $X_1$ or $X_2$ and then two vertices of $Y_1$ or of $Y_2$ are on $P$, which is not possible because $Y_1$ and $Y_2$ are cliques.  This means that $P$ will pass through $X_1$ and $X_2$, and thus through $H$. Because $X_1\cup Y_1$ and $X_2\cup Y_2$ are cliques, we can safely mimic this part  of $P$ by the path from $u_1$ to $u_2$ in the subgraph of $H$ that we did not remove.
From now on we assume that $H$ contains at least one terminal vertex.

Below we split Case 2 in a number of subcases. In these subcases we solve the problem for  a graph that is either $H$ or the graph obtained from $H$ by adding a new vertex
adjacent to every vertex in $X_1$ and/or a new vertex adjacent to every vertex in $X_2$. Hence, this graph is  isomorphic to a (not necessarily proper)
induced subgraph of $J$. As such,  this subroutine takes $n^{O(1)}$ time, due to Claim~1.

\medskip
\noindent
{\bf Case 2a.} Both $X_1$ and $X_2$ contain a terminal vertex.\\
Because $X_1$ and $X_2$ are cliques, $X_1$ and $X_2$ each contain exactly one terminal vertex. Let $u_h$ be the terminal vertex of $X_h$ for $h=1,2$.

If  there is a terminal pair $(s_j,t_j)$ with one of $s_j,t_j$ in $V(H)\setminus (X_1\cup X_2)$ and the other one in $V(F)\setminus (Y_1\cup Y_2)$, then we return {\tt No}. The reason is that in this case any $s_jt_j$-path must either pass through $X_1$ or through $X_2$. Because $X_1$ and $X_2$ are cliques each containing a terminal vertex, this is not possible.
From now on we assume that such a terminal pair $(s_j,t_j)$ does not exist. 

\medskip
\noindent
{\bf Case 2ai.} $u_1$ and $u_2$ represent terminals of the same pair.\\
Let this pair be $(s_i,t_i)$.
Because $(G,S)$ is an independent instance, $s_i$ and $t_i$ are not represented by the same vertex. Hence,
we may assume without loss of generality that $u_1$ represents $s_1$ and that $u_2$ represents $t_2$. 
The fact that $(G,S)$ is an independent instance also implies that $u_1$ and $u_2$ are not adjacent.

\medskip
\noindent
{\bf Case 2ai-1.} All partners of the terminals in $T_{u_1}$ and all partners of the terminals in $T_{u_2}$ belong to $H$.\\
Then we first solve the problem for $H$. If we find a solution, then we reduce to $F-(Y_1\cup Y_2)$.
Otherwise, the $s_it_i$-path of any solution for $G$ only contains vertices from $F$ besides $u_1$ and $u_2$. In particular, such a path would use one vertex from $Y_1$ and one vertex from $Y_2$ (which may be the same vertex in case $Y_1$ and $Y_2$ have a common vertex). We now proceed as follows.

If $T_{u_1}=\{s_i\}$ and $T_{u_2}=\{t_i\}$, then no neighbors of $u_1$ in $H$ and no neighbor of $u_2$ in $H$ can be used as an inner vertex of some solution path.
Hence, we first solve the problem for $H-(N_H[u_1]\cup N_H[u_2])$. 
If the answer is {\tt No}, then $G$ has no solution. Otherwise, we reduce to $F+u_1+u_2$. 

If $T_{u_1}=\{s_i\}$ and $|T_{u_2}|=2$, then no neighbor of $u_1$ in $H$ is used as an inner vertex of some solution path in any solution for $G$, whereas 
one neighbor $w$ of $u_2$ in $H$ will be used as an inner vertex, because $|T_{u_2}|=2$.
However, such a vertex $w$ cannot be in $X_2$, because then it would still be adjacent to the inner vertex of the $s_it_i$-path that is in $Y_2$, as $X_2\cup Y_2$ is a clique. 
Hence, we first solve the problem for $H-(N_H[u_1]\cup (X_2\setminus \{u_2\})$.  
If the answer is {\tt No}, then $G$ has no solution. Otherwise, we  reduce to $F+u_1+u_2$. 

If $|T_{u_1}|=|T_{u_2}|=2$, then no vertex of $X^*=(X_1\setminus \{u_1\})\cup (X_2\setminus \{u_2\})$ can be used as an (inner) vertex of some solution path in any solution for 
$G$ for the same reason as in the previous case. Hence, we first solve the problem for $H-X^*$. If the answer is {\tt No}, then $G$ has no solution. Otherwise, we reduce to $F+u_1+u_2$. 

\medskip
\noindent
{\bf Case 2ai-2.} All partners of the terminals of one of $T_{u_1}, T_{u_2}$, say of $T_{u_1}$, belong to $H$, while $T_{u_2}$ 
contains a terminal, the partner of which is not in $H$.\\
Then $T_{u_2}$ consists of exactly two terminals. Suppose that $s_j\in T_{u_2}$ for some $j\neq i$.
Then the $s_it_i$-path of any solution for $G$ uses no vertices from $F$, whereas the $s_jt_j$-path of any solution for $G$ uses only vertices from $F$ besides $u_2$.
Moreover, an $s_it_i$-path cannot use a vertex from $X_2$ as an inner vertex, because such a vertex would be adjacent to the inner vertex of the $s_jt_j$-path that is in $Y_2$, and $X_2\cup Y_2$ is a clique.
Hence we first solve the problem for $H-(X_2\setminus\{u_2\})$. If the answer is {\tt No}, then $G$ has no solution.
Otherwise, we reduce to $F-Y_1+u_2$. 
 
\medskip
\noindent 
{\bf Case 2ai-3.} Both $T_{u_1}$ and $T_{u_2}$ contain a terminal, the partner of which does not belong to $H$.\\
Because two terminal pairs do not coincide, we find that the other terminals represented by $u_1$ and $u_2$ belong to a different pair.
Hence, we may without loss of generality assume that
$s_h$ with $h\neq i$ is the other terminal represented by $u_1$, and that $s_j$ with $j\notin \{h,i\}$ is the other terminal represented by $u_2$.

By the same arguments as in the Case 2ai-2, we can first solve the problem for $H-X^*$, where $X^*=(X_1\setminus \{u_1\})\cup (X_2\setminus \{u_2\})$.
If the answer is {\tt No}, then $G$ has no solution. 
Otherwise, we reduce to $F+u_1+u_2$.

\medskip
\noindent
{\bf Case 2aii.} $u_1$ and $u_2$ do not represent terminals of the same pair.\\
We say that $u_i$ with $1\leq i\leq 2$ is {\it mixed} if a partner of one terminal represented by $u_i$ is   in $H$, and a partner of one terminal represented by $u_i$ is in $F$.
If all partners of the terminals represented by $u_i$ are in $H$, then we say that $u_i$ is
 {\it $H$-homogeneous}.
If all partners of the terminals represented by $u_i$ are in $F$, then we say that $u_i$ is {\it $F$-homogeneous}.
In this way, we can distinguish a number of cases, where we use arguments that we already used in the previous cases.

Suppose that $u_1$ and $u_2$ are both 
$H$-homogeneous.
Then we first solve the problem for $H$. If the answer is {\tt No}, then $G$ has no solution.
Otherwise, we reduce to $F-(Y_1\cup Y_2)$.

Suppose that $u_1$ and $u_2$ are both $F$-homogeneous.
Then we first solve the problem for $H-(N_H[u_1]\cup N_H[u_2])$. If the answer is {\tt No}, then $G$ has no solution.
Otherwise, we reduce to $F+u_1+u_2$ .
 
 Suppose that one of $u_1,u_2$, say $u_1$, is $H$-homogeneous, whereas $u_2$ is $F$-homogeneous.
 Then we first solve the problem for $H-N_H[u_2]$. If the answer is {\tt No}, then $G$ has no solution.
 Otherwise, we reduce to $F-Y_1+u_2$.
 
 Suppose that one of $u_1,u_2$, say $u_1$, is $H$-homogeneous, whereas $u_2$ is mixed.
Then we first solve the problem for $H-(X_2\setminus \{u_2\})$. If the answer is {\tt No}, then $G$ has no solution.
 Otherwise, we reduce to $F-Y_1+u_2$.
 
 Suppose that one of $u_1,u_2$, say $u_1$, is $F$-homogeneous, whereas $u_2$ is mixed.
 Then we first solve the problem for $H-N_H[u_1]-(X_2\setminus \{u_2\})$. If the answer is {\tt No}, then $G$ has no solution.
 Otherwise, we reduce to $F+u_1+u_2$.
 
Suppose that both $u_1$ and $u_2$ are mixed.
Then we first solve the problem for $H-X^*$,  where $X^*=(X_1\setminus \{u_1\})\cup (X_2\setminus \{u_2\})$.
If the answer is {\tt No}, then $G$ has no solution.
 Otherwise, we reduce to $F+u_1+u_2$.

\medskip
\noindent
This completes  Case 2a. Note that we never branched in this case.

\medskip
\noindent
{\bf Case 2b.} Only one of the sets $X_1,X_2$ contains a terminal vertex.\\
We assume without loss of generality that $X_1$ contains a terminal vertex $u$, and consequently, that $X_2$ contains no terminal vertex.
If the vertices of $V(H)\setminus \{u\}$ represent two or more terminals whose partners are in $F$, then $G$ has no solution.
From now on, we assume that there is at most one terminal that is represented by a
 vertex in $V(H)\setminus \{u\}$ and that has its partner in $F$.

We now start to branch for the first time. We do this into four directions. In the first three directions
we check whether $G$ has a solution that contains no vertex from $X_2$, $Y_1$, or $Y_2$, respectively. 
In these cases we may remove $X_2$, $Y_1$, or $Y_2$, respectively, from $G$ and return to Case 1.
In the remaining branch we check whether $G$ has a solution in which a solution path uses a vertex from each of the sets $X_2$, $Y_1$, and $Y_2$; note that these three vertices will be inner vertices of one or more solution paths. This is the branch we analyze below.

We borrow the notions of $u$ being $F$-homogeneous, $H$-homogeneous, or mixed from Case 2aii.
Recall that, because our instance is independent, $u$ does not represent two terminals of the same pair. Hence, we may denote the terminals in $T_u$ by
$s_i$, or by $s_i,s_j$ depending on whether $u$ represents one or two terminals.
We also use the following notations. 
Let $F^*$ denote the graph obtained from $F$ by adding a new vertex $u_1'$ adjacent to all vertices of $Y_1$, a new vertex $u_2'$, and a new vertex adjacent to all vertices of $Y_2$ and to $u_2'$.
Let $H^*$ denote the graph obtained from $H$ by removing $N_H[u]$ from $H$ and adding a vertex $v'$ adjacent to 
all vertices in $X_2$. 
Let $H'$ denote the graph obtained from $H$ by removing $X_1\setminus \{u\}$ from $H$ and adding a vertex $v'$ adjacent to 
all vertices in $X_2$. 
Note that $H^*$ and $H'$  are induced subgraphs of $J$, and thus Claim 1 can be used.

We distinguish the following subcases. 

\medskip
\noindent
{\bf Case 2bi.} $u$ is $F$-homogeneous.\\
Recall that in this stage of the algorithm we investigate whether $G$ has a solution, such that $X_1$, $X_2$, $Y_1$, and $Y_2$ each contain a vertex that will be used on a solution path. Then in this case, such a solution must contain a solution path that starts in $u$ and uses a vertex from $Y_1$.
Since this solution path cannot end in $H$, it cannot use a vertex from $X_2$.
Then there must exist some other solution path that uses a vertex from $X_2$ and a vertex from $Y_2$. This solution path cannot have both end-vertices in $H$ due to 
the solution path starting from $u$.
Hence, $H$ must contain a terminal vertex representing a terminal whose partner is not in $H$; otherwise we can stop considering this branch.
Let $(s_h,t_h)$ be this terminal, where we assume that $s_h$ is represented by a terminal vertex in $H$. So, the $s_ht_h$-solution path will use a vertex from $X_2$ and a vertex from $Y_2$.

\medskip
\noindent
{\bf Case 2bi-1.} $T_u=\{s_i\}$.\\
Then the $s_it_i$-path uses a vertex from $Y_1$. We now proceed as follows. 
We let $v'$ represent a new terminal $t_h'$ that is the new partner of $s_h$ in $H^*$.
Then we solve the problem for $H^*$. If the answer is {\tt No}, then we stop considering this branch.
Otherwise, we let $u_1', u_2'$ represent  new 
terminals $s_i'$ and $s_h'$, respectively, that form the new terminals for $t_i$ and $t_h$, respectively, in $F^*$, and we reduce to $F^*$.

\medskip
\noindent
{\bf Case 2bi-2.} $T_u=\{s_i,s_j\}$.\\ 
We must branch into two directions, as either only the $s_it_i$-path or only the $s_jt_j$-path can use a vertex from $Y_1$.

Suppose that only the $s_it_i$-path will use a vertex from $Y_1$ (the other case is symmetric). 
Then $h=j$, because $u$ is $F$-homogeneous.
We remove $s_i$ from the set of terminals in $H'$, and
we let $v'$ represent a new terminal $t_h'$ that is the new partner of $s_h$ in $H'$.
Then we solve the problem for $H'$. If the answer is {\tt No}, then we stop considering this branch.
Otherwise we let $u_1', u_2'$ represent  new 
terminals $s_i'$ and $s_h'$, respectively, that form the new terminals for $t_i$ and $t_h$, respectively, in $F^*$, and we reduce to $F^*$.

\medskip
\noindent
{\bf Case 2bii.} $u$ is $H$-homogeneous.\\
In this case one of the solution paths starting in $u$ consecutively passes through $Y_1$, $Y_2$, and $X_2$. This path does not use any vertex
from $N_H(u)$, as otherwise it would not be induced.  If $H$ contains a vertex that represents a terminal of which the partner is not in $H$, we stop with considering this branch. 
Otherwise, we proceed as follows.

First suppose that $T_u=\{s_i\}$. 
We let $v'$ represent a new terminal $s_i'$ that is the new partner of $t_i$ in $H^*$. 
Then we solve the problem for $H^*$. If the answer is {\tt No}, then we stop considering this branch. Otherwise, we let $u_1', u_2'$ represent  new 
terminals $s_i'$ and $t_i'$, respectively, that form a new terminal pair in $F^*$, and we reduce to $F^*$.

Now suppose that $T_u=\{s_i,s_j\}$. We branch into two directions.
In the first branch, we remove $s_i$ from the set of terminals in $H'$, and we let  $v'$ represent a new terminal $s_i'$ as the new partner of $t_i$ in $H'$.
Then we solve the problem for $H'$. If the answer is {\tt No}, then we stop considering this branch.
 Otherwise, we let $u_1', u_2'$ represent  new 
terminals $s_i'$ and $t_i'$, respectively, that form a new terminal pair in $F^*$, and we reduce to $F^*$.
In the second branch, we do the same thing as in the first branch, but with $(s_{j},t_{j})$ instead of $(s_{i},t_{i})$.

\medskip
\noindent
{\bf Case 2biii.} $u$ is mixed.\\
Suppose that $t_i$ is represented by a terminal vertex in $H$, and hence, $t_j$ is represented by a terminal vertex in $F$.
In that case the $s_it_i$-path belongs to $H$ and the $s_jt_j$-path belongs to $F+u$. As such, the latter path cannot use a vertex from $X_2$.
Because the solution path that uses a vertex from $X_2$ must also be the solution path that uses a vertex from $Y_2$, this solution path cannot have both end-vertices in $H$. 
Hence, $H$ must contain a terminal vertex representing a terminal whose partner is not in $H$; otherwise we can stop considering this branch.
Let $(s_h,t_h)$ be this terminal, where we assume that $s_h$ is represented by a terminal vertex in $H$, and consequently, $t_h$ is represented by a terminal vertex in $F$.

We now proceed as follows. We remove $s_j$ from $T_u$.
We let $v'$ represent a new terminal $t_h'$ as the new partner of $s_h$ in $H'$.  
Then we solve the problem for $H'$. If the answer is {\tt No}, then we stop considering this branch.
Otherwise, we let $u_1', u_2'$ represent  new 
terminals $s_j'$ and $s_h'$, respectively, that form a new terminal pair in $F^*$, and we reduce to $F^*$.

\medskip
\noindent
This completes Case 2b. Note that we branched into at most five directions.

\medskip
\noindent
{\bf Case 2c.} Neither $X_1$ nor $X_2$ contains a terminal vertex.\\
Recall that in this stage of the algorithm $H$ is assumed to contain at least one terminal vertex. 
We branch in five directions. In the first four directions, we check whether $G$ has a solution that contains no vertex from $X_1$, $X_2$, $Y_1$, $Y_2$, respectively.
In these cases we may remove $X_{1}$, $X_2$, $Y_1$, or $Y_2$, respectively, from $G$ and return to Case 1.
In the remaining branch we check whether $G$ has a solution in which a solution path uses a vertex from each of the sets $X_1$, $X_2$, $Y_1$, and $Y_2$. Note that these four vertices will be inner vertices of one or more solution paths. This is the branch that we analyze below.

We say that a terminal that is represented by a vertex in $H$ but whose partner is represented by a vertex in $F$ is {\it unpaired} in $H$.
If at least three terminals are unpaired in $H$, then $G$ has no solution. This leads to three subcases, in which we use the following additional notations.
Let $H''$ be the graph obtained from $H$ by adding a new vertex $v_1'$ adjacent to all vertices in $X_1$ and a new vertex $v_2'$ adjacent to all vertices in $X_2$. 
Note that $H'$ is isomorphic to $J$.
We let $F^*$ denote the graph obtained from $F$ by adding a new vertex $u_1'$, a new vertex adjacent to all vertices of $Y_1$ and to $u_1'$, a new vertex $u_2'$, and a new vertex adjacent to all vertices of $Y_2$ and to $u_2'$.

\medskip
\noindent
{\bf Case 2ci.} No terminal is unpaired in $H$.\\
We first verify the following. Let $v_1'$ and $v_2'$ represent new terminals $s_h'$ and $t_h'$ that form a new terminal pair in $H''$. We then solve
the problem for $H''$. 

First suppose that $H''$ has a solution. Then we remove all vertices of $H$ from $G$ except the vertices from a shortest path from a vertex $u_1\in X_1$
to a vertex $u_2\in X_2$. We may do so, because the resulting graph $G'$ has a solution if and only if $G$ has a solution, as we just confirmed that we can always ``fit'' the solution paths
between terminals in $H$. 

Now suppose that $H''$ has no solution. Because we investigate whether $G$ has a solution such that $X_1$, $X_2$, $Y_1$, and $Y_2$ each contain a vertex that is used on a solution path, we must now check whether $G$ has a solution that contains a solution path that starts in a vertex of $H$, passes through the four aforementioned sets in
order $X_1,Y_1,Y_2,X_2$ or in order $X_2,Y_2,Y_1,X_1$, and finally ends in a vertex of $H$ again.

For each terminal pair $(s_i,t_i)$ that is represented in $H$, we check whether $H''$ has a solution, after letting $v_1',v_2'$ represent new terminals $t_i',s_i'$, respectively, that are the new partners of $s_i$ and $t_i$, respectively, in $H''$.
We also check the possibility if $H''$ has a solution after letting $v_1',v_2'$ represent new terminals $s_i',t_i'$, respectively, that are the new partners of $t_i$ and $s_i$, respectively, in $H''$. If the answer is {\tt No} for both possibilities for all terminal pairs represented in $H$, then we stop considering this branch.
Otherwise, we reduce to $F^*$ after letting $u_1',u_2'$ represent new terminals $s_h',t_h'$, respectively, that form a new terminal pair in $F^*$.
Note that we did not do any further branching in this subcase, that is, we either stop this branch, or we continue with graph $G'$ or $F^*$.

\medskip
\noindent
{\bf Case 2cii.} Exactly one terminal is unpaired in $H$.\\
Let $s_i$ be this terminal. Then the $s_it_i$-path must pass through $X_i$ and $Y_i$ for $i=1$ or $i=2$. 
However, then it is not possible for any other solution path to pass through $X_j$ and $Y_j$ for $j\neq i$.
Hence, we do not have to consider this case in our branching algorithm.

\medskip
\noindent
{\bf Case 2ciii.} Exactly two terminals are unpaired in $H$.\\
Because these two terminals are unpaired, we may denote them by $s_i$ and $s_j$, respectively. Note that they may be represented by the same vertex.
We further branch in two directions.

First, we check whether $H''$ has a solution after letting $v_1',v_2'$ represent new terminals $t_i',t_j'$, respectively, that are the new partners of $s_i$ and $s_j$, respectively, in $H''$. If the answer is {\tt No}, then we stop considering this branch.
Otherwise, we reduce to $F^*$ after letting $u_1',u_2'$ represent new terminals $s_i',s_j'$, respectively, that are the new partners of $t_i$ and $t_j$ in $F^*$.

Second, we check the possibility if $H''$ has a solution after letting $v_1',v_2'$ represent new terminals $t_j',t_i'$, respectively, that are the new partners of $s_j$ and $s_i$, respectively, in $H''$. If the answer is {\tt No}, then we stop considering this branch.
Otherwise, we reduce to $F^*$ after letting $u_1',u_2'$ represent new terminals $s_j',s_i'$, respectively, that are the new partners of $t_j$ and $t_i$ in $F^*$.

\medskip
\noindent
This completes Case 2c, which was the last case in our analysis. Note that we branched into at most six directions in Case 2c. 

\medskip
After our branching algorithm we have either found in 
$k^6n^{O(1)}$ 
time that $G$ has no solution, or a set ${\cal G}$ of at most $6^k$ graphs.
This can be seen as follows. First,  we processed each strip in 
$n^{O(1)}$ time.
Second,  our algorithm neither recomputed a strip structure from scratch nor created any new strips when going through the iterations. Moreover,
for each stripe $(J,Z)$ with no terminal vertices in $J\setminus Z$, the algorithm did not branch
at all, 
and for each strip $(J,Z)$ with terminal vertices in $J\setminus Z$, 
it 
branched into at most six directions. Hence, the corresponding search tree of our branching algorithm has depth $k$ and at most $6^k$ leaves.

Because we only removed vertices from $G$, we find that every graph in ${\cal G}$ has at most $n$ vertices.
Because we only removed terminal pairs from $S$ or replaced a terminal pair by another terminal pair, we find that every graph in ${\cal G}$ has at most $k$ terminal pairs. Moreover, for each graph $G'\in {\cal G}$, it holds that the neighborhood of each of its vertices is the disjoint union of at most two cliques. This is true, because
every stripe corresponds to a path of three vertices and every spot corresponds to a vertex that is in exactly two maximal cliques, which are disjoint, because of the twin-freeness. 
Hence, $G'$ is a line graph by Lemma~\ref{l-linegraph}. This completes the proof of Lemma~\ref{l-strip}.
\qed
\end{proof}

\subsection{Line Graphs}\label{s-linegraphs}

In this section we state and prove Lemma~\ref{l-line} used in Step 7.

\begin{lemma}\label{l-line}
The 
{\sc Induced Disjoint Paths}
problem can be solved in $g(k)n^6$ time for line graphs on $n$ vertices and with $k$ terminal pairs, where $g$ is a function that only depends on $k$.
\end{lemma}

\begin{proof}
Let $G$ be a line graph with terminal pairs $(s_1,t_1),\ldots,(s_k,t_k)$. Let $H$ be the preimage of $G$, which we can obtain in linear time due to Lemma~\ref{l-preimage}. Recall that by definition there is a bijection between vertices of $G$ and edges of $H$. Let $e_{v} \in E(H)$ denote the edge corresponding to vertex $v \in V(G)$. Furthermore, given a vertex $h \in V(H)$, let $V_{h}$ denote the set of vertices in $G$ corresponding to the edges of $H$ that are incident to $h$. Observe that $V_{h}$ is a clique in $G$.

We first preprocess the instance in $O(k^{2}n + n^{2})$ time using the rules of Lemma~\ref{l-pre} in order to obtain an independent instance. Observe that the class of line graphs is closed under vertex deletion, and thus Rules~1, 2, and 3 of Lemma~\ref{l-pre} preserve membership of the class of line graphs. It remains to verify that Rule~4, which potentially removes edges, is also safe. This can be seen as follows. Consider two adjacent terminal vertices $u,v \in V(G)$. Then $e_{u}$ and $e_{v}$ are both incident to a vertex $h \in V(H)$. Since Rule~1, 2, and 3 have been applied, every vertex of $V_{h}$ is a terminal vertex in $G$. As Rule~4 will thus remove all edges between vertices of $V_{h}$, we can update the preimage by deleting $h$ and replacing each incident edge $f$ with an edge $f'$ to a new vertex $h_{f}$. It follows that the rules of Lemma~\ref{l-pre} preserve membership of the class of line graphs.

By abuse of notation, we still use $G$ and $(s_1,t_1),\ldots,(s_k,t_k)$ to denote the graph and the terminal pairs, respectively, of the preprocessed instance, and $H$ to denote the preimage of $G$. Consider a terminal vertex $x$ of $G$ and its corresponding edge $e_{x}=u_iv_i$ in $H$. If $x$ represents one terminal, then we choose one of $u_i,v_i$, say $u_{i}$. Then we let $u_i$ represent the terminal represented by $x$ in $G$ and remove all neighbors of $u_i$ except $v_i$ from $H$. If $x$ represents two terminals, then they must be from distinct terminal pairs, say $(s_{i},t_{i})$ and $(s_{j},t_{j})$. We may assume that $x$ represents $s_{i}$ and $s_{j}$. Then we replace the edge $e_{x}$ with the edges $u_{i}a$ and $bv_{i}$, where $a$ and $b$ are new vertices, and consider the two possible assignments of $s_{i},s_{j}$ to $a,b$ for which each of $a,b$ represents exactly one terminal. 
Because we have at most $2k$ terminal vertices in $G$,  this leads to at most $2^{2k}$ new graphs $H'$. 

We claim that $G$ has a solution if and only if one of the new graphs $H'$ with corresponding terminal pairs forms a {\tt Yes}-instance of {\sc Disjoint Paths}; in that case
we also say that a graph $H'$ has a solution.
Our claim 
can be seen as follows. First, we observe that mutually induced paths in a line graph  are in one-to-one correspondence with vertex-disjoint paths in its preimage.
Because we consider both options for picking an end-vertex of each ``terminal edge'' in $H$, this means that a solution for $G$ can be translated to a solution for at least one of the graphs $H'$.  Second, by letting a terminal edge be the only edge incident to the chosen end-vertex, we guarantee that a solution for a graph $H'$ can be translated to a solution for $G$.
 
We are left to apply  Lemma~\ref{t-RS} at most $2^{2k}$ times. Note that $H$ contains $O(n^2)$ vertices and that each call to Lemma~\ref{t-RS} takes
$h(k)\, |V_H|^3$ time, where $h(k)$ is a function that only depends on $k$. Hence, the total running time is $g(k)\, n^6$ for $g(k)=2^{2k}h(k)$. 
This completes the proof of Lemma~\ref{l-line}.\qed
\end{proof}

Lemma~\ref{l-line} completes the proof of Theorem~\ref{t-main}. A similar result has also been used by Fiala \etal\cite{FKLP12}, but we had to do a more careful running time analysis in order to show our fpt-result.

\subsection{Parameterized Complexity of Related Problems}
Theorem~\ref{t-main} implies a similar result for the problems {\problemIAC}, {\problemIAP}, and {\problemIAT} for claw-free graphs.

\begin{corollary}\label{c-main}
The problems {\problemIAC}, {\problemIAP}, and {\problemIAT} are fixed-parameter tractable for claw-free graphs when parameterized by $k$.
\end{corollary}

\begin{proof}
First we consider the {\problemIAC} problem.
Let $G$ be a claw-free graph with a set $U=\{u_1,\ldots,u_k\}$ of $k$ specified vertices. 
Recall that {\problemIAC} can be solved in polynomial time for any fixed $k$, as shown by Fiala \etal\cite{FKLP12}. 
Hence, 
we may assume that $k \geq 3$. We fix an order of the vertices in $U$, say $U$ is ordered as $u_1,\ldots,u_k$. We define terminal pairs $(s_i,t_i)=(u_i,u_{i+1})$ for $i=1,\ldots, k-1$ and $(s_k,t_k)=(u_k,u_1)$. Then we apply Theorem~\ref{t-main}. If this does not yield a solution, then we consider a different order of the vertices of $U$ until we considered them all. This adds an extra factor of $k!$ to the running time of the fpt-algorithm of Theorem~\ref{t-main}.

The proof for the {\problemIAP} problem uses the same arguments as for the {\problemIAC} problem when $k \geq 3$. The only difference is that we do not have a terminal pair $(s_k,t_k)$. Finally, recall that for claw-free graphs the {\problemIAP} problem is equivalent to the {\problemIAT} problem.
\qed
\end{proof}

\section{Induced Topological Minors}\label{s-hard}
In this section we investigate to what extent we can apply Theorem~\ref{t-main} to detect induced containment relations. We first show the following result.

\begin{theorem}\label{t-anchored}
The {\sc Anchored Induced Topological Minor} problem is fixed-parameter tractable for 
pairs $(G,H)$, where $G$
is a claw-free  graph, 
$H$ is an (arbitrary) graph, 
and $|V(H)|$ is the parameter.
\end{theorem}

\begin{proof}
Let $G$ be a claw-free graph with $k$ specified vertices ordered as $u_1,\ldots,u_k$ for some integer $k$. Let $H$ be an arbitrary $k$-vertex graph, whose vertices are ordered as $x_1,\ldots,x_k$. For each isolated vertex $x_i\in V(H)$, we define a terminal pair $(u_i,u_i)$.
For each edge $x_ix_j\in E(H)$, we define a terminal pair $(u_i,u_j)$. This leads to a set of terminal pairs $T=\{(s_1,t_1),\ldots,(s_\ell,t_\ell)\}$, where $\ell$ is the number of edges and isolated vertices of $H$. 
Because $H$ has no multiple edges, no two terminal pairs in $G$ coincide. Hence the created set of terminal pairs has Properties 1 and 2.  
Then $G$ contains an induced subgraph isomorphic to a subdivision of $H$ such that the isomorphism maps $u_i$ to $x_i$ for $i=1,\ldots, k$ if and only if
$G$ contains a set of $\ell$ mutually induced paths $P_1,\ldots,P_\ell$, such that $P_j$ has end-vertices $s_j$ and $t_j$ for $j=1,\ldots,\ell$.  Because $H$ is fixed, $\ell$ is a constant. Hence, we may apply Theorem~\ref{t-main}, and the result follows. \qed
\end{proof}

Observe that, using Theorem~\ref{t-anchored}, it is easy to solve the {\sc Induced Topological Minor} problem for pairs $(G,H)$ (where $G$ is a claw-free graph) in $O(f(|V(H)|)\ n^{|V(H)|+O(1)})$ time. We simply guess the anchors of  the topological minor in $n^{|V(H)|}$ time and then run the algorithm of Theorem~\ref{t-anchored} in $O(f(|V(H)|)\ n^{O(1)})$ time, for some function $f$.
However, this algorithm is hardly an improvement over the existing 
$n^{O(|V_H|)}$-time
algorithm for the {\sc Induced Topological Minor} problem for pairs $(G,H)$ (where $G$ is a claw-free graph) that was developed by Fiala \etal\cite{FKLP12}. We show in fact that any substantial improvement on this result is unlikely, since we prove below that the problem is \W[1]-hard.

\begin{theorem}\label{t-notan}
The {\sc Induced Topological Minor} problem is \W$[1]$-hard for 
pairs $(G,H)$ where $G$ and $H$ are line graphs, and $|V(H)|$ is the parameter.
\end{theorem}

\begin{proof}
We give a reduction from the {\sc Clique} problem, which asks whether a graph has a clique of size at least $k$. 
This problem is \W$[1]$-complete when parameterized  by~$k$ (cf.  Downey and Fellows~\cite{DF99}).

Let $G$ be a graph and $k$ an integer; we may assume without loss of generality that $k\geq 4$. We claim that $G$ has a clique of size $k$ if and only if 
$L(G)$ contains $L(K_{k})$ as an induced topological minor.

First suppose that $G$ has a clique of size $k$. Then it contains a graph $G'$ isomorphic to $K_k$ as an induced subgraph. 
In $L(G)$ we remove all vertices that correspond to edges incident with at least one vertex in $V(G)\setminus V(G')$. This leads to an induced subgraph in $L(G)$ that
is isomorphic to $L(K_k)$. It remains to observe that any induced subgraph of a graph is also an induced topological minor of that graph.

Now suppose that $L(G)$ contains $L(K_k)$ as an induced topological minor. Then there exists a sequence $S$ of vertex deletions and vertex dissolutions 
that modifies $L(G)$ into $L(K_k)$.
We claim that $S$ only consists of vertex deletions.
In order to obtain a contradiction, suppose that $S$ contains at least one vertex
dissolution. 
We may without loss of generality assume that all vertex deletions in $S$ occur before the vertex dissolutions in $S$.
Let $F$ be the graph obtained from $L(G)$ after these vertex deletions. 
Because the class of line graphs is closed under vertex deletions, $F$ is a line graph.
Moreover, by construction, $F$ is a subdivision of $L(K_{k})$.

By our assumption, $F$ contains at least one vertex $e$ of degree two that must be dissolved in order to obtain a graph isomorphic to $L(K_{k})$. Let $f$ be one of the two neighbors of $e$ in $F$.
Note that $L(K_{k})$ is the union of $k$ cliques $S_1,\ldots,S_k$ of size $k-1\geq 3$ 
that pairwise share exactly one vertex
in such a way that every vertex of $L(K_k)$ belongs to exactly two cliques $S_i$ and $S_j$. 
This implies that $ef$ must be an edge inside one of these cliques. However, then $f$ is the center of a claw. Because $H$ is a line graph, this is not possible. Hence, $S$ contains no vertex dissolutions, and consequently, $F$ is isomorphic to $L(K_{k})$. Because the vertex deletions in $S$ translate to edge deletions in $G$, we then find that $K_{k}$ is a subgraph of $G$. In other words, $G$ contains a clique of size $k$.
This completes the proof of Theorem~\ref{t-notan}. \qed
\end{proof}

It is less clear to what extent induced linkages can be used to find some fixed induced minor in a claw-free graph. So far, limited  progress has been made on the {\sc $H$-Induced Minor} problem for claw-free graphs, although more polynomial cases are known for this graph class than for general graphs~\cite{FKP12}.

\section{Conclusions}\label{s-con}

We showed that the {\sc Induced Disjoint Paths} problem is fixed-parameter tractable in $k$ for claw-free graphs. As a consequence, we also proved that the
problems {\problemIAC}, {\problemIAP}, and {\problemIAT} are fixed-parameter tractable in $k$, and that
the same result applies to {\sc Anchored Induced Topological Minor} when parameterized by the number of vertices in the target graph $H$.
We also showed that our results cannot be applied to the {\sc Induced Topological Minor} problem, which turned out to be \W[1]-hard even on line graphs. In this section, we show that  our result for the {\sc Induced Disjoint Paths} problem is also tight from two other perspectives, and we state some open problems.

It is natural to ask whether our results generalize to $K_{1,\ell}$-free graphs for $\ell\geq 4$. We show that this is unlikely.

\begin{proposition}\label{p-k14}
The problems {\sc 2-Induced Disjoint Paths}, {\sc 2-in-a-Cycle}, and {\sc $3$-in-a-Path} are \NP-complete even for $K_{1,4}$-free graphs.
\end{proposition}

\begin{proof}
Derhy and Picouleau~\cite{derhy2009} proved that {\sc $3$-in-a-Path}  is \NP-complete even for graphs with maximum degree at most three.
L\'{e}v\^{e}que \etal\cite{LLMT09} proved that {\sc $2$-in-a-Cycle} is \NP-complete even for graphs with maximum degree at most three and terminals of degree two. From this, it follows immediately that {\sc $2$-Induced Disjoint Paths} is \NP-complete for graphs with maximum degree at most three, because we can subdivide the two edges incident with each terminal and then place terminals $s_1, s_2, t_1, t_2$ on the four newly created vertices. It remains to observe that graphs of maximum degree at most three are  $K_{1,4}$-free.
\qed
\end{proof}

The next step would be to try to construct a polynomial kernel for {\sc Induced Disjoint Paths} restricted to claw-free graphs.
However, we show that this is not likely even for line graphs.
This follows from the work of Bodlaender, Thomass\'e, and Yeo~\cite{BTY11}, who showed that {\sc Disjoint Paths} has no polynomial kernel when parameterized by $k$, unless  \NP\ $\subseteq$\ co\NP$/$poly, together with the fact that an instance $(G,(s_1,t_1),\ldots,(s_k,t_k))$
of {\sc Disjoint Paths} can be translated to an instance $(L(G),(s_1',t_1'),\ldots,(s_k',t_k'))$ as follows. 
For each vertex in $G$ that represent $p\geq 1$ terminals we introduce a new vertex only adjacent to this vertex, and we let this new vertex represent the $p$ terminals instead. Then the added edges become the vertices that represent the terminals in $L(G)$.

\begin{proposition}\label{p-nopolykernel}
The
{\sc Induced Disjoint Paths}
problem restricted to line graphs has no polynomial kernel when parameterized by $k$, unless \NP\ $\subseteq$\ \textup{co}$\NP/$\textup{poly}.
 \end{proposition}
 
The question whether the same result as in Proposition~\ref{p-nopolykernel} holds for {\problemIAC} and {\problemIAP} restricted to line graphs is open.

Instead of improving our result for the {\sc Induced Disjoint Paths} problem, we could also work towards solving a more general problem. In the definition of induced disjoint paths, we explicitly disallowed duplicate terminal pairs, that is, the set of terminal pairs is not a multiset. 

If we generalize to allow duplicate terminal pairs, then we can solve the {\sc $k$-Induced Disjoint Paths} problem for claw-free graphs in polynomial time for fixed $k$ as follows.
In a nutshell, we may assume without loss of generality that no vertex represents more than two terminals (otherwise we have
a no-instance). Then, for any two  terminal pairs $(s_i,t_i)$ and $(s_j,t_j)$ with $s_i=s_j$ and $t_i=t_j$, we replace
$(s_j,t_j)$ by a new pair $(s_j',t_j)$ where $s_j'$ is a neighbor of $s_i=s_j$.
This only adds an extra $O(n)$ factor to the running time for each pair of coinciding terminal pairs,
because we just have to explore all possible choices of such a neighbor. 

Determining  the parameterized complexity of the general case is still an open problem.
As a partial result towards answering this question, we consider the variation of  {\sc Induced Disjoint Paths} where all terminal pairs coincide.
For $k=2$, this problem is equivalent to the {\sc $2$-in-a-Cycle} problem, which is \NP-complete~\cite{Bi91,Fe89} for general graphs and solvable 
in~$O(n^2)$ time for $n$-vertex planar graphs~\cite{MRSS1994}. 
For claw-free graphs, 
recall
that no terminal vertex can represent more than two terminals in any {\tt Yes}-instance. Hence the problem can be reduced to the {\sc $2$-in-a-Cycle} problem, which is polynomial-time
solvable
on claw-free graphs~\cite{FKLP12}.

Finally, we note that there may be other natural parameters for the problems considered. For example, Haas and Hoffmann~\cite{HH06} consider the {\sc 3-in-a-Path} problem and  prove \W$[1]$-completeness for general graphs if the parameter is the length of an induced path that is a solution for {\sc $3$-in-a-Path}.

\end{document}